\documentclass[11pt]{article}

\usepackage[T1]{fontenc}			
\usepackage[utf8]{inputenc}
\usepackage[english]{babel}
\usepackage{amsmath, amsfonts, amssymb, amsthm, amscd, mathtools}
\usepackage{algorithmic}
\usepackage{bbm}
\usepackage{xspace} 
\usepackage{hyperref}
\usepackage[square,numbers]{natbib}	
\usepackage{xurl}
\usepackage[normalem]{ulem}
\usepackage{algorithm}

\usepackage{pgfplots}
\usepackage{subfig}
\pgfplotsset{compat=1.7}
\usepackage{graphicx}
\usepackage{eurosym}
\usepackage{authblk}

\usepackage{multimedia}
\usepackage[shortlabels]{enumitem} 

  \usepackage{pstricks, pst-node,pst-text,pst-3d,pst-all,pst-3dplot}
\usepackage{tikz}
\usetikzlibrary{arrows} 
\usepackage{subfig}
\usepackage{float}
\usepackage{caption}

\usepackage{comment}

\newcommand{\bbE}{{\ensuremath{\mathbb E}} }
\newcommand{\bbF}{{\ensuremath{\mathbb F}} }

\newcommand{\bbQ}{{\ensuremath{\mathbb Q}} }

\newcommand{\cB}{{\ensuremath{\mathcal B}} }
\newcommand{\cC}{{\ensuremath{\mathcal C}} }

\newcommand{\cF}{{\ensuremath{\mathcal F}} }

\newcommand{\cL}{{\ensuremath{\mathcal L}} }

\newcommand{\cN}{{\ensuremath{\mathcal N}} }

\newcommand{\cR}{{\ensuremath{\mathcal R}} }
\newcommand{\cS}{{\ensuremath{\mathcal S}} }
\newcommand{\cT}{{\ensuremath{\mathcal T}} }

\newcommand{\dd}{{\ensuremath{\mathrm d}} }
\newcommand{\de}{{\ensuremath{\mathrm e}} }

\newcommand{\dC}{{\ensuremath{\mathrm C}} }

\newcommand{\R}{\mathbb{R}}

\newcommand{\N}{\mathbb{N}}

\newcommand{\ind}{\ensuremath{\mathbf{1}}}

\DeclarePairedDelimiter{\abs}{\lvert}{\rvert}

\DeclarePairedDelimiterX{\inprod}[2]{\langle}{\rangle}{#1, #2}

\newlength\myindent
\numberwithin{equation}{section}

\newtheorem{theorem}{Theorem}[section]

\newtheorem{corollary}[theorem]{Corollary}

\newtheorem{lemma}[theorem]{Lemma}

\newtheorem{proposition}[theorem]{Proposition} 
\newtheorem{remark}[theorem]{Remark}

\newtheorem{assumption}[theorem]{Assumption}

\hypersetup{pdftitle={Short-rate models with stochastic discontinuities:\\
a PDE approach},pdfauthor={Alessandro Calvia, Marzia De Donno, Chiara Guardasoni, Simona Sanfelici},%
pdftex,pdfstartview={XYZ null null 1.4},colorlinks=true,linkcolor=blue,urlcolor=blue}

\begin{document} 

	\title{\textbf{Short-rate models with stochastic discontinuities:\\
 a PDE approach}}

\author[1]{Alessandro Calvia \thanks{alessandro.calvia@polimi.it}}

\author[2]{Marzia De Donno \thanks{marzia.dedonno@unicatt.it}}

\author[3]{Chiara Guardasoni \thanks{chiara.guardasoni@unipr.it}}

\author[4]{Simona Sanfelici \thanks{simona.sanfelici@unipr.it}}

\affil[1]{\footnotesize Department of Mathematics, Politecnico di Milano. Member of the INdAM-GNAMPA group.}

\affil[2]{\footnotesize  School of Banking, Finance and Insurance, Universit\`a Cattolica del Sacro Cuore, Milano.}

\affil[3]{\footnotesize Department of Mathematical Physical and Computer Sciences, Universit\`a di Parma. Member of INdAM-GNCS group.} 

\affil[4]{\footnotesize Department of Economics and Management, Universit\`a di Parma. Member of INdAM-GNCS group.} 
\maketitle

\textwidth=160 mm \textheight=225mm \parindent=8mm \frenchspacing
\vspace{3 mm}

\begin{abstract}
With the reform of interest rate benchmarks, interbank offered rates (IBORs) like LIBOR have been replaced by risk-free rates (RFRs), such as the Secured Overnight Financing Rate (SOFR) in the U.S. and the Euro Short-Term Rate (\euro STR) in Europe. These rates exhibit characteristics like jumps and spikes which correspond to specific market events, driven by regulatory and liquidity constraints. To capture these characteristics, this paper considers a general short-rate model that incorporates discontinuities at fixed times with random sizes. Within this framework, we introduce a PDE-based approach for pricing interest rate derivatives and  establish, under suitable assumptions, a Feynman-Kač representation for the solution. For affine models, we derive (quasi) closed-form solutions, while for the general case, we develop numerical methods to solve the resulting PDEs. 

\medskip 


\medskip 


\end{abstract}

\section{Introduction}

In recent years, the global financial system has undergone a significant transformation in the way interest rates are benchmarked. 
The discontinuation of the London Interbank Offered Rate (LIBOR){, culminating in the \textit{LIBOR funeral} -- as termed in~\cite{klingler2021life} -- with the publication of the last LIBOR settings\footnote{See the \textit{FCA announcement on future cessation and loss of representativeness of the LIBOR
benchmarks}, available at \url{https://www.fca.org.uk/publication/documents/future-cessation-loss-representativeness-libor-benchmarks.pdf}, and the \textit{Decisions on US dollar LIBOR}, available at \url{https://www.fca.org.uk/publication/feedback/fs23-2.pdf}.} at the end of June~2023,} has led to the widespread adoption of alternative overnight rates, known as Risk-Free Rates (RFRs). {Examples are the \textit{Secured Overnight Financing Rate} (SOFR) in the United States of America, the \textit{Euro Short-Term Rate} (€STR) in Europe; the \textit{Sterling Overnight Index Average} (SONIA) in the United Kingdom, the \textit{Tokyo Overnight Average Rate} (TONA) in Japan, and the \textit{Swiss Average Rate Overnight} (SARON) in Switzerland. For a description and an analysis from the economic point of view of the first three risk-free rates, see, e.g., \cite{klingler2021life}.}

{The transition from LIBOR and other IBORs to risk-free rates (RFRs) was driven by structural weaknesses and credibility issues that became evident after the 2008 global financial crisis, when repeated LIBOR manipulations were discovered. The central issue was that the publication of LIBOR rates, which were intended to measure the unsecured short-term interbank borrowing cost, was based on banks’ self-reported estimates rather than actual market transactions. This survey-based fixing mechanism made the rate vulnerable to manipulation, which occurred both for profit-driven purposes related to derivative positions and, during the 2008 financial crisis, to conceal banks’ liquidity and credit risk conditions.
The LIBOR scandal revealed not only misconduct by individual institutions, but also a fundamental flaw in the design of the benchmark itself, which undermined its reliability as a reference rate for financial markets. In response, regulatory authorities initiated a comprehensive reform\footnote{See, e.g., the document \textit{Reforming Major Interest Rate Benchmarks}, Financial Stability Board, July 2014, available at \url{https://www.fsb.org/2014/07/r_140722/}, and the subsequent \textit{Progress report}, December 2019, available at \url{https://www.fsb.org/2019/12/reforming-major-interest-rate-benchmarks-progress-report-2/}.} of benchmark interest rates, aiming to restore market integrity and financial stability.
Risk-free rates were introduced as a more robust alternative because they are based on observable, high-volume overnight transactions and entail virtually zero credit risk. Their transaction-based nature makes them more transparent and far less susceptible to manipulation.}

The shift toward overnight RFRs as the primary observable in financial markets has significant implications for the modeling of interest rate dynamics. Traditional term structure models, which relied on longer-maturity instruments, must now be recalibrated to incorporate the properties of these \mbox{short-maturity} rates. As a result, there has been a renewed interest in short-rate models, adapted to the new framework dictated by the post-LIBOR era. 

{ A first approach to modeling term rates based on RFRs was introduced in \cite{lyashenko2019looking}, who propose the generalized Forward Market Model (FMM) as an extension of the classical Libor Market Model, allowing for the joint modeling of forward-looking and backward-looking rates within a single, unified framework. The FMM has subsequently been developed in several contributions, including \cite{lyashenko2020, lopezsalas24}. Other authors have instead developed extensions of classical short-rate models to the RFR setting; see, among others, \cite{rutkowski2021pricing, mercurio2018multi, russo2023transition, skov2021dynamic}. Additional references can be found in \cite{fontanaetal2024:termstruct, liu2024risk}.}

A key feature in the dynamics of RFRs is the presence of jumps or spikes, particularly around central bank announcements and regulatory constraints. In the context of SONIA and SOFR,  \cite{BackwellHayes2022}  develop jump-diffusion models that differentiate between expected and unexpected jumps, while \cite{Andersen02082024} model spikes occurring both at predictable and inaccessible times, with specific applications to SOFR-linked derivatives.
However, the analysis by \cite{anbiletal2020} shows that such discontinuities often occur at deterministic dates, especially those aligned with central bank policy meetings or structural market features. Jumps driven by monetary policy decisions were first incorporated into interest rate modeling by \cite{piazzesi2005bond}, and later by \cite{kim2014jumps}, who coined the term \emph{stochastic discontinuities} to describe jumps of random magnitude occurring at predetermined times.
In the context of RFR modeling, several contributions have proposed tractable frameworks that incorporate such features. Notably,  \cite{Brace2024, DeGenaro2018, gellert2021short, heitfield2019inferring, liu2024risk, skov2021dynamic} focus primarily on the overnight benchmarks SOFR and SONIA, and introduce a variety of modeling approaches, ranging from reduced-form models to affine term structure models and Markov chain approximations.  
More recently, \cite{Angelini2025Modeling}  extend this line of research  to the euro area by proposing an affine short-rate model for the joint evolution of EURIBOR and \euro STR. Moreover, \cite{fusai2024} present a novel tree approach to pricing derivatives linked to new RFR benchmarks in a discrete-time setting, inspired by the analogy with the pricing of Asian options.
A more general continuous-time HJM-based framework is developed in \cite{fontana2020term, fontanaetal2024:termstruct}, where stochastic discontinuities  are modeled through affine semimartingales. As an  example, the authors extend the classical Hull-White model by introducing jumps at deterministic dates with independent random sizes. They derive the characteristic function of the resulting process and use martingale methods to price interest rate derivatives, including zero-coupon bonds  and caplets (see also \cite{Fontana2023}).  { Finally,   \cite{fontanaPS} propose an extension of the CIR model with  stochastic discontinuities, which, under suitable assumptions, allows for both upward and downward jumps, while preserving the affine structure. }
 
 In this paper, we focus on this class of models, that is, short-rate models that incorporate stochastic discontinuities. Our first contribution is to propose an alternative pricing methodology for derivatives, widely used in short-rate frameworks: the partial differential equation (PDE) approach. One key advantage of this method is its independence from the specific functional form of the short-rate dynamics. This allows us to consider a general class of models, encompassing not only affine processes but also more general  structures.
As in \cite{fontanaetal2024:termstruct}, we assume that the num\'eraire evolves continuously between roll-over dates and exhibits jumps at roll-over times. Consequently, our framework features two sources of discontinuity: one stemming from the short-rate process itself and another from the num\'eraire.
We show that, under suitable assumptions, the price of a European-style derivative can be obtained by solving a PDE between jump times and piecing the solutions at jump times using boundary conditions dictated by no-arbitrage arguments. Furthermore, we provide a rigorous analysis of existence, uniqueness, and regularity of the solution to the associated Cauchy problem. Classical PDE results are not directly applicable here due to the unboundedness of the potential term and the exponential growth of the terminal condition. Nevertheless, we succeed in proving a Feynman–Kač representation formula for the solution.
We then apply this \mbox{PDE-based} framework to the affine setting. Under general assumptions, we derive a closed-form solution for the price of a zero-coupon bond (ZCB). Additionally, we obtain a closed-form expression for the price of a European call option on a ZCB in a Vasicek model with stochastic discontinuities, assuming normally distributed jumps.

Our analysis shows that even within the affine class, closed-form pricing formulas may not always be attainable and, when available, may involve cumbersome computations. As a further contribution,  we therefore propose a numerical approach for solving  the PDE. Two different numerical methods are investigated and compared against the available analytical formulas. The first, a semi-analytic method, yields highly accurate results but entails greater computational cost and limited applicability, as it requires the explicit knowledge of the Green’s function. The second, a finite-difference scheme, while slightly less accurate than the semi-analytic approach, still provides a very high level of precision. Moreover, it stands out for its remarkable flexibility and computational efficiency, making it well-suited for practical applications.
In summary, this work highlights the potential of the PDE approach in handling more general interest rate models and derivative contracts. Its flexibility and robustness make it a valuable tool, particularly when closed-form solutions are unavailable or difficult to derive.

The paper is organized as follows. 
In Section 2, we present the short-rate model incorporating stochastic discontinuities. Section 3 derives the PDE formulation for pricing general interest rate derivatives. Section 4 investigates the well-posedness of the PDE and the regularity of its solution, and proves that, under suitable assumptions, it admits a Feynman-Kač representation. In Section 5, we apply the framework to affine models and obtain closed-form solutions in specific cases. Section 6 presents the numerical methods for pricing interest rate derivatives also when analytical solutions are not available. Finally, Section 7 concludes the paper. 

\section{The model}\label{sec:model}
Let $(\Omega, \cF, \bbF, \mathbb{Q})$ be a complete probability space, with filtration $\bbF \coloneqq (\cF_t)_{t\ge 0}$ satisfying the usual assumptions. 
We assume that the probability space is rich enough to support a real standard Brownian motion $W = (W_t)_{t \geq 0}$ and a \mbox{pure-jump} process $J = (J_t)_{t \geq 0}$ given by
\begin{equation}\label{jumpprocess}
J_t \coloneqq \sum_{i=1}^{M} \xi_i \ind_{[s_i, +\infty)} (t) \, , \qquad t \geq 0 \, ,
\end{equation}
where, for a fixed $M \in \N$, $0 < s_1 < \cdots < s_M < +\infty$ are given deterministic jump times and $\xi_1, \dots, \xi_M$ are real random variables (see below for the financial meaning of these quantities), which are assumed to be independent of $W$ and $\cF_{s_i}$-measurable, for each $i = 1, \dots, M$.

{For any given $t \geq 0$ and $x \in \R$, we consider the stochastic differential equation (SDE)
\begin{equation}\label{RFR}
\begin{dcases}
  \dd \rho_s = \mu(s, \rho_s) \, \dd s + \sigma(s, \rho_s) \, \dd W_s + \dd J_s \, , \quad s \geq t \, , \\
  \rho_t = x \, .
  \end{dcases}
\end{equation}
}%
The drift and volatility coefficients $\mu \colon [0, +\infty) \times \mathbb{R} \to \mathbb{R}$ and $\sigma \colon [0, +\infty) \times \mathbb{R} \to \mathbb{R}^+$, are such that a path-wise unique strong solution to~\eqref{RFR} exists, for any given initial condition.
Whenever needed, we will use the notation $\rho^{t,x}$ to denote the solution to~\eqref{RFR} starting at time $t \geq 0$ from $x \in \R$, i.e., $\rho^{t,x}_t = x$.
{We observe that if} random variables $\xi_1, \dots, \xi_M$ are independent, then it can be shown that $\rho$ is a Markov process. This is because the solution to SDE~\eqref{RFR} can be obtained by piecing together the solutions of the corresponding SDE without jumps, with the appropriate initial conditions, between jump times of process $J$.

\begin{remark}\label{rem:noSDEassumpt}
    We do not assume \textit{a priori} any specific conditions on the coefficients of SDE~\eqref{RFR} because, on the one hand, this will not have an influence on the general results that we will prove in Section~\ref{PDEformulation}, and, on the other hand, because from Section~\ref{App:ExUniq} onward we will discuss specific models, arising from particular choices of these coefficients.
\end{remark}

{Process $\rho$ models the evolution of the Risk Free Rate (RFR for short). We associate to it the money market account $S^0$, acting as numéraire, which} is obtained by investing in the overnight rate. For its modeling, we adopt the approach of~\cite{fontanaetal2024:termstruct}, where the authors propose a highly general framework that bridges two perspectives: on one side, the classical interest rate model, where the num\'eraire represents the continuous-time limit of a roll-over strategy; and on the other side, a more realistic approach, where the num\'eraire is piecewise constant with jumps at roll-over dates (see \cite{fontanaetal2024:termstruct} for details).  To this end, given the set of \emph{roll-over dates} $\cT \coloneqq \{t_1, \ldots, t_N \}$, with $0 < t_1 < \cdots < t_N < +\infty$, the num\'eraire is defined as: 
\begin{equation} \label{numeraire}
 S^0_t \coloneqq \exp \left(\int_0^t \rho_u \, \eta(\dd u)\right) = \exp\left\{\int_0^t\rho_u \, \dd u  + \sum_{n=1}^N \rho_{t_n} \ind_{[t_n, +\infty)} (t)\right\} \, , \quad t \geq 0 \, ,
\end{equation}
where measure $\eta$ is defined\footnote{As usual, $\delta_t$ denotes the Dirac probability measure concentrated
at $t \in \R^+$.} as $\eta (A) \coloneqq \int_A \dd u + \sum_{n=1}^N \delta_{t_n}(A)$, for $ A \in \cB(\mathbb{R}^+)$.

{Throughout this paper we work under the following standing assumption on the market model.
\begin{assumption}
    \mbox{}
    \begin{enumerate}[label = \arabic*.]
        \item There exists a market for zero coupon bonds (ZCB for short) with maturity $T$, for any $T > 0$.
        \item\label{hp:risk-neutral} Probability measure $\bbQ$ is a risk-neutral measure, that is, for any fixed maturity $T > 0$ the ZCB price process $(P(t,T))_{t \in [0,T]}$ satisfies the condition $P(T,T) = 1$, and is such that the discounted price process $\left(\frac{P(t,T)}{S^0_t}\right)_{t \in [0,T]}$ is a $\bbQ$-martingale.
    \end{enumerate}
\end{assumption}
Point~\ref{hp:risk-neutral} of the above assumption implies that we work under a martingale modeling approach, which is customary in short-rate models. That is, we directly specify the dynamics of the short rate under a risk-neutral measure. Assuming the existence of such a measure entails that our model is arbitrage-free, thus all ZCB and interest-rate derivatives must be consistently priced within our setting.
}

{We conclude this section with some comments on the model for the short-rate dynamics that we introduced in~\eqref{RFR}.}

In this model, which extends the framework introduced in Section 4.3 of \cite{fontanaetal2024:termstruct}, the RFR may jump at fixed times $s_1, \dots, s_M$, which we collect in the set $\cS \coloneqq \{s_1, \ldots, s_M \}$ of \emph{expected jump dates.}
At these times, the RFR has jumps of magnitude $\xi_i$, $i = 1, \dots, M$. 
{%
Note that, although in our model we consider a finite number of dates where jumps in the RFR or the num\'eraire  may occur, our results can be extended to the case in which both $\cS$ and  $\cT$ are countable sets. 
The choice to focus on a finite number of jump and roll-over dates is justified by two reasons. First, since we aim to price derivatives with a finite payoff horizon, only a finite number of dates are relevant for the pricing procedure. Second, from a calibration perspective, relevant dates are typically known at most 12 months in advance, making it unnecessary to consider a model with a long-term horizon.

This model is sufficiently general to encompass, as special cases, affine models such as the Vasicek and Cox-Ingersoll-Ross (CIR) models, as well as log-normal models like the Black-Karasinski model and the class of CEV models (\cite{Chanetal92:CEV, marshrosenfeld83:CEV}).
Moreover, this model effectively captures both the spikes and jumps typically observed in the dynamics of the RFR. Specifically, low volatility on jump dates results in distinct jumps, while strong mean reversion may give rise to sharp spikes. 

In general, the observed trajectories of the  { RFR} (see, e.g., Figure~1 in \cite{anbiletal2020}) suggest that spikes in the rate can be reasonably modeled using a continuous distribution for their size $\xi_i$. In contrast, jumps typically take values that are integer multiples of a fixed quantity (such as a basis point), which makes a discrete random variable  a more suitable choice to model their size. Indeed, in Section~\ref{sec:simulation}, where we simulate short-rate trajectories under the assumption that the RFR follows a Hull–White process with jumps, we model the jump component using either a Gaussian or a two-point discrete distribution.}

\section{Pricing of derivatives: a PDE formulation}\label{PDEformulation}
In this section we show that, under appropriate assumptions, the price of a European derivative contract written on the spot interest rate $\rho$, with maturity $T > 0$ and terminal payoff $H(\rho_T)$ can be expressed as a function of time and of the RFR, which solves an appropriate PDE.
Conversely, in Section~\ref{App:ExUniq} we also prove a Feynman-Kač representation formula under suitable assumptions.

Throughout the rest of the paper we assume the following. 
\begin{assumption}\label{hp:PDE}
    \mbox{}
    \begin{enumerate}[label=\arabic*.]
        \item The random variables $\xi_1, \dots, \xi_M$, describing the jump sizes of the RFR, are independent, with distribution $Q_1, \dots, Q_M$, respectively. In addition, they are independent of the Brownian motion $W$.
        \item For any $i = 1, \dots, M$, $\xi_i$ is $\cF_{s_i}$-measurable.
        \item\label{hp:PDE:existuniq} SDE~\eqref{RFR} admits a unique strong solution $\rho^{t,x}$, for any given initial condition $(t,x) \in [0,T] \times \R$, and $\rho_u^{t,x}$ has support $\R$, for all $u \in (t,T]$.
        \item\label{hp:PDE:H} The contract function $H \colon \R \to \R$ is measurable and such that
        \begin{equation*}
            \bbE\biggl[\de^{-\int_0^T\rho_s \, \eta(\dd s)} \abs{H(\rho_{T})}\biggr] < +\infty \, .
        \end{equation*}
    \end{enumerate}
\end{assumption}
\noindent As we observed in Section~\ref{sec:model}, Assumption~\ref{hp:PDE} entails that the RFR process $\rho$ is Markovian.

\begin{remark}
\label{rem:weaker_assumpt}
{
    The results of this Section remain valid under the following weaker version of point~\ref{hp:PDE:existuniq} in Assumption~\ref{hp:PDE} (cf. \citep[Chapter~5]{filipovic2009term}). 
\begin{enumerate}[label=\arabic*a.]\setcounter{enumi}{2}
    \item\label{hp:PDE:existuniqweak} 
    SDE~\eqref{RFR} admits a unique strong solution $\rho^{t,x}$, for any given initial condition $(t,x) \in [0,T] \times A$, and $\rho_u^{t,x}$ has support $A$, for all $u \in (t,T]$, where $A \subseteq \R$ is a closed interval with non-empty interior.
\end{enumerate}
    We choose to focus on the case in which the state space of process $\rho$ is the whole real line, i.e., $A=\R$, for simplicity and also for another main reason. In Section~\ref{App:ExUniq} we study the well-posedness of a family of partial differential equations (PDEs), related to the pricing function of the derivative, on the domain $[\tau_0, \tau] \times \R$, where $[\tau_0, \tau) \subset [0,T)$ is a given interval (see Equation~\ref{eq:genericPDE}). As we will discuss later in that Section, to establish existence and uniqueness of a classical solution to the PDE we cannot rely on standard results, because of the presence of a potential coefficient, which is unbounded in $\R$. Indeed, classical results require that this coefficient is bounded from below, an assumption that in our setting would be verified if the above interval $A$ were also bounded from below. This would be the case, e.g., of $A = [0,+\infty)$ (implying a non-negative RFR), which could be analyzed with standard results, instead.
    For the sake of completeness, we will briefly comment whenever needed what changes if we substitute point~\ref{hp:PDE:existuniq} with point~\ref{hp:PDE:existuniqweak} in Assumption~\ref{hp:PDE}. 
}
\end{remark}
Since we are working in an arbitrage-free setting, the price process $V = (V_t)_{t \in [0,T]}$  
of the derivative contract is given by the usual risk-neutral valuation formula
\begin{equation*}
    V_t = \bbE\biggl[\de^{-\int_t^T\rho_s \, \eta(\dd s)} H(\rho_{T}) \biggm| \cF_t\biggr] \, , \quad t \in [0,T] \, .
\end{equation*}
Exploiting the Markovianity of process $\rho$, we can write
\begin{equation*}
    V_t = f(t, \rho_t) \coloneqq \bbE\biggl[\de^{-\int_t^T\rho_s \, \eta(\dd s)} H(\rho_{T}) \biggm| \rho_t\biggr] \, , \quad t \in [0,T] \, ,
\end{equation*}
where the \emph{pricing function} $f \colon [0,T] \times \R \to \R$ is defined as
\begin{equation}\label{eq:pricefunction}
    f(t,x) \coloneqq \bbE\biggl[\de^{-\int_t^T\rho^{t,x}_s \, \eta(\dd s)} H(\rho^{t,x}_{T}) \biggr] \, , \quad t \in [0,T], \, x \in \R \, .
\end{equation}
 
Let us introduce the set $\cR \coloneqq (\cS \cup \cT) \cap [0,T]$ of \textit{relevant dates}, i.e., jump times $r_1 < r_2 < \cdots < r_K$ before maturity $T$, where $K$ is the cardinality of set $\cR$. More precisely, setting $r_0 \coloneqq 0$, we have that
\begin{equation*}
    r_k \coloneqq \min\{t \in (r_{k-1}, T] \colon t \in \cS \cup \cT\}, \quad k \in \{1, \dots, K\} \, .
\end{equation*}
Note that $K \leq M+N$, where $M$ is the number of expected jump dates of the RFR $\rho$ and $N$ is the number of roll-over dates of the numéraire $S^0$. Strict inequality may hold even if all these dates are before maturity $T$, since there may be common jump times between processes $\rho$ and $S^0$.
Moreover, it may also happen that $r_K = T$.

\begin{remark}\label{rem:no_left_cont}
    Note that we do not expect the pricing function $f(t,x)$ to be left-continuous with respect to $t$. This is precisely due to the stochastic discontinuities at the relevant dates. However, it has left-hand limits (still with respect to $t$), as the price process $V$ is càdlàg. This is also clearly seen in the explicit formula for bond prices given in~\citep[Proposition~4.6]{fontanaetal2024:termstruct}, in the case of a Gaussian distribution of jumps at expected jump dates.
\end{remark}

In the following we use the notation $f(t^-,x) \coloneqq \lim_{s \uparrow t} f(s,x)$, for any $(t,x) \in [0,T] \times \R$. 
Let us introduce the operator $\cL$, defined for any $\varphi \in \dC^{1,2}((0,T) \times \R)$ as
\begin{equation}\label{L}
    \cL \varphi(t,x) \coloneqq \partial_t \varphi(t, x) + \mu(t,x) \partial_x \varphi(t, x) + \dfrac 12 \sigma^2(t,x) \partial_{xx}^2 \varphi(t, x), \quad (t,x) \in (0,T) \times \R \, .
\end{equation}
We also introduce the integer-value random measure $m^\cS$ given by
\begin{equation}\label{eq:mS}
m^\cS(\dd s \, \dd z) \coloneqq \sum_{m=1}^M \delta_{(s_m, \xi_m)}(\dd s \, \dd z) \, ,
\end{equation} 
with compensator
\begin{equation}\label{eq:muS}
    \mu^\cS(\dd s \, \dd z) \coloneqq \sum_{m=1}^M \delta_{s_m}(\dd s) \, Q_m(\dd z) \, .
\end{equation}

\begin{theorem}
    Suppose that $f \in \cC([0,T]\setminus \cR \times \R) \cap \cC^{1,2}((0,T) \setminus \cR \times \R)$ and that, for any $x \in \R$, the map $t \mapsto f(t,x)$ is right-continuous. Assume, moreover, that $f$ is such that the process $M^f \coloneqq (M^f_t)_{t \in [0,T]}$ given by
    \begin{multline*}
        M^f_t \coloneqq \int_0^t \sigma(s, \rho_s) \partial_x f(s, \rho_s) \, \dd W_s + \int_0^t \ind_{s \in \cS \setminus \cT} \int_\R \left[f(s, \rho_{s^-} + z) - f(s^-, \rho_{s^-})\right] [m^\cS - \mu^\cS](\dd s \dd z)
        \\ + \int_0^t \ind_{s \in \cS \cap \cT} \int_\R \left[\de^{-(\rho_{s^-} + z)} f(s, \rho_{s^-} + z) - f(s^-, \rho_{s^-})\right] [m^\cS - \mu^\cS](\dd s \dd z)
    \end{multline*}
    is a local martingale. 
    Then, under Assumption~\ref{hp:PDE}, $f$ satisfies the backward  { PDE}s
    \begin{equation}\label{eq:lastPDE}
    \begin{dcases}
        \cL f(u, x) - x f(u, x) = 0 \, , & (u,x) \in [r_K, T) \times \R \, , \\
        f(T,x) = H(x) \, , & x \in \R \, ,
    \end{dcases}
\end{equation}
and, for all $k \in \{1, \dots, K\}$,
\begin{equation} \label{finalSystem}
    \begin{dcases}
        \cL f(u, x) - x f(u, x) = 0 \, , & (u,x) \in [r_{k-1}, r_k) \times \R \, , \\
        f(r_k^-, x) = \de^{-x} f(r_k, x) \, , & x \in \R, \, r_k \in \cT \setminus \cS\, , \\
        f(r_k^-,x) = \int_\R f(r_k, x+z) \, Q_{m(k)}(\dd z) \, , & x \in \R, \, r_k \in \cS \setminus \cT \, , \\
        f(r_k^-, x) = \int_\R \de^{-(x+z)} f(r_k, x+z) \, Q_{m(k)}(\dd z) \, , & x \in \R, \, r_k \in \cS \cap \cT \, ,
    \end{dcases}
\end{equation}
where, for each $k \in \{1, \dots, K\}$ such that $r_k \in \cS$, $m(k)$ is the unique index in $\{1, \dots, M\}$ for which $r_k = s_{m(k)}$.
\end{theorem}

\begin{proof}
Fix $(t,x) \in [0,T) \times \R$ and consider the RFR starting at $(t,x)$, i.e., $\rho^{t,x} = (\rho^{t,x}_u)_{u \in [t,T]}$.
Let us define the process
\begin{equation}\label{eq:R}
    R^{t,x}_u \coloneqq \int_t^u \rho^{t,x}_s \, \eta(\dd s) = \int_t^u \rho^{t,x}_s \, \dd s + \sum_{i=1}^N \rho^{t,x}_{t_i} \ind_{\{t < t_i \leq u\}} \, , \quad u \in [t,T] \, .
\end{equation}
We observe that, for any $u \in [t,T]$,
\begin{equation*}
    R^{t,x}_{u^-} = \int_t^u \rho^{t,x}_s \, \dd s + \sum_{i=1}^N \rho^{t,x}_{t_i} \ind_{\{t < t_i < u\}} \, ,
\end{equation*}
whence, for any $n = 1, \dots, N$ such that $t_n \in [t,T]$,
\begin{align*}
    R^{t,x}_{t_n} = \int_t^{t_n} \rho^{t,x}_s \, \dd s + \sum_{i=1}^{n-1} \rho^{t,x}_{t_i} + \rho^{t,x}_{t_n} = R^{t,x}_{t_n^-} + \rho^{t,x}_{t_n} \, .
\end{align*}

In the following, to unburden the notation, we will avoid specifying the dependence on the initial condition of processes $\rho$ and $R$, unless necessary.

Applying the It\^o's product formula (see, e.g.,~\citep[Theorem 14.1.1]{cohen:stochcalculus}) to the process $\de^{-\int_t^u\rho_s \, \eta(\dd s)} f(u,\rho_u) = \de^{-R_u} f(u,\rho_u)$, $u \in [t,T]$, we have that

\begin{equation}\label{eq:product_formula}
    \de^{-R_u} f(u, \rho_u) = f(t, x) + \int_t^u \de^{-R_{s^-}} \, \dd f(s, \rho_s) + \int_t^u f(s^-, \rho_{s^-}) \, \dd\bigl(\de^{-R_s}\bigr) + \int_t^u \dd [\de^{-R_\cdot}, f(\cdot, \rho_\cdot)]_s \, .
\end{equation}

The first integral can be computed observing that if $(t,u]$ does not contain any relevant date, then
\begin{equation*}
    f(u, \rho_u) = f(t,x) + \int_t^u \cL f(s, \rho_s) \, \dd s + \int_t^u \sigma(s, \rho_s) \partial_x f(s, \rho_s) \, \dd W_s \, ;
\end{equation*}
if instead $(t,u] \cap \cR \neq \emptyset$, then, denoting by $\ell$ and $L$, respectively, the first and the last indices in $\{1, \dots, K\}$ such that $r_k \in (t,u]$, we have that
\begin{align}\label{eq:f_decomp}
    f(u, \rho_u) 
    &= f(t,x) + [f(r_{\ell}, \rho_{r_{\ell}}) - f(t, x)] + \sum_{k = \ell + 1}^{L} [f(r_k, \rho_{r_k}) - f(r_{k-1}, \rho_{r_{k-1}})] + [f(u, \rho_u) - f(r_{L}, \rho_{r_{L}})] \notag
    \\
    &= f(t,x) + [f(r_{\ell}, \rho_{r_{\ell}}) - f(r_{\ell}^-, \rho_{r_{\ell}^-})] + [f(r_{\ell}^-, \rho_{r_{\ell}^-}) - f(t, x)] \notag
    \\
    &\quad + \sum_{k = \ell + 1}^{L} \left\{[f(r_k, \rho_{r_k}) - f(r_k^-, \rho_{r_k^-})] + [f(r_k^-, \rho_{r_k^-}) - f(r_{k-1}, \rho_{r_{k-1}})]\right\} + [f(u, \rho_u) - f(r_{L}, \rho_{r_{L}})] \notag
    \\
    &= f(t,x) + \sum_{k = \ell}^{L} [f(r_k, \rho_{r_k}) - f(r_k^-, \rho_{r_k}^-)] \notag
    \\
    &\quad + [f(r_{\ell}^-, \rho_{r_{\ell}^-}) - f(t, x)] + \sum_{k = \ell + 1}^{L} [f(r_k^-, \rho_{r_k}^-) - f(r_{k-1}, \rho_{r_{k-1}})] + [f(u, \rho_u) - f(r_{L}, \rho_{r_{L}})] \, .
\end{align}
Note that if $u = r_{L}$, then the last summand disappears. Since in the last line the left-continuous version (with respect to time) of $f$ is used, the usual Itô's formula can be applied to each term. Indeed, we can directly apply it to the last summand, and (to deal with the preceding two) we can define, for any $k = \ell, \dots, L-1$, the auxiliary functions
\begin{equation*}
    g_\ell(\tau, y) \coloneqq
    \begin{dcases}
        f(\tau, y), & (\tau,y) \in [t, r_\ell) \times \R \, , \\
        f(r_\ell^-,y), & \tau = r_\ell, \, y \in \R \, ,
    \end{dcases}
    \qquad
    g_k(\tau, y) \coloneqq
    \begin{dcases}
        f(\tau, y), & (\tau,y) \in [r_k, r_{k+1}) \times \R \, , \\
        f(r_{k+1}^-,y), & \tau = r_{k+1}, \, y \in \R \, ,
    \end{dcases}
\end{equation*}
and the family of stochastic processes $X^{\tau_0,y}$, each solving, for any given $(\tau_0,y) \in [t,T) \times \R$, the SDE
\begin{equation*}
    \begin{dcases}
        \dd X_s = b(s, X_s) \, \dd s + \sigma(s, X_s) \, \dd W_s, & s \in (\tau_0, T] \, , \\
        X_{\tau_0} = y \, .
    \end{dcases}
\end{equation*}
Note that $\rho_s = X_s^{t,x}$, for all $s \in [t, r_\ell)$, and (by the Markov property) $\rho_s = \rho_s^{r_k, \rho_{r_k}^{t,x}} = X_s^{r_k, \rho_{r_k}^{t,x}}$, for all $s \in [r_k, r_{k+1})$, $k = \ell, \dots, L-1$. Moreover,
\begin{equation*}
    f(r_{\ell}^-, \rho_{r_{\ell}^-}) - f(t, x) = g_\ell(r_\ell, X_{r_\ell}^{t,x}) - g_\ell(t,x) \, ,
    \quad
    f(r_k^-, \rho_{r_k}^-) - f(r_{k-1}, \rho_{r_{k-1}}) = g_k(r_k, X_{r_k}^{r_{k-1}, \rho_{r_{k-1}}^{t,x}}) - g_k(r_{k-1}, \rho_{r_{k-1}}^{t,x})\, ,
\end{equation*}
for any $k = \ell, \dots, L-1$. In addition, each of the functions $g_\ell, \dots, g_{L-1}$ is continuous in its domain and of class $\cC^{1,2}$ in the interior. Therefore, by Itô's formula,
\begin{align*}
     &\mathop{\phantom{=}} [f(r_{\ell}^-, \rho_{r_{\ell}^-}) - f(t, x)] + \sum_{k = \ell + 1}^{L} [f(r_k^-, \rho_{r_k}^-) - f(r_{k-1}, \rho_{r_{k-1}})] + [f(u, \rho_u) - f(r_{L}, \rho_{r_L})]
     \\
     &= [g_\ell(r_{\ell}, X_{r_\ell}^{t,x}) - g_\ell(t, x)] + \sum_{k = \ell + 1}^{L} [g_k(r_k,X_{r_k}^{r_{k-1}, \rho_{r_{k-1}}^{t,x}}) - g_k(r_{k-1}, \rho_{r_{k-1}}^{t,x})] + [f(u, \rho_u) - f(r_{L}, \rho_{r_L}^{t,x})]
     \\
     &= \int_t^{r_\ell} \cL g_\ell(s, X_s^{t,x}) \, \dd s + \int_t^{r_\ell} \sigma(s, X_s^{t,x}) \partial_x g_\ell(s, X_s^{t,x}) \, \dd W_s 
     \\
     &\qquad + \sum_{k = \ell + 1}^{L} \left\{ \int_{r_k}^{r_{k+1}} \cL g_k(s, X_s^{r_{k-1}, \rho_{r_{k-1}}^{t,x}}) \, \dd s + \int_{r_k}^{r_{k+1}} \sigma(s, X_s^{r_{k-1}, \rho_{r_{k-1}}^{t,x}}) \partial_x f(s, X_s^{r_{k-1}, \rho_{r_{k-1}}^{t,x}}) \, \dd W_s \right\}
     \\
     &\qquad + \int_{r_L}^u \cL f(s, \rho_s) \, \dd s + \int_{r_L}^u \sigma(s, \rho_s) \partial_x f(s, \rho_s) \, \dd W_s
     \\
     &= \int_t^u \cL f(s, \rho_s) \, \dd s + \int_t^u \sigma(s, \rho_s) \partial_x f(s, \rho_s) \, \dd W_s \, .
\end{align*}
Hence, substituting this result in~\eqref{eq:f_decomp}, we get
\begin{align}\label{eq:f_semimart}
    f(u, \rho_u) 
    &= f(t,x) + \int_t^u \cL f(s, \rho_s) \, \dd s + \int_t^u \sigma(s, \rho_s) \partial_x f(s, \rho_s) \, \dd W_s + \sum_{k = \ell}^{L} [f(r_k, \rho_{r_k}) - f(r_k^-, \rho_{r_k}^-)] \notag
    \\
    &= f(t,x) + \int_t^u \cL f(s, \rho_s) \, \dd s + \int_t^u \sigma(s, \rho_s) \partial_x f(s, \rho_s) \, \dd W_s \notag
    \\
    &\qquad + \sum_{t < s \leq u} \ind_{s \in \cT \setminus \cS} [f(s, \rho_s) - f(s^-, \rho_s)] + \sum_{t < s \leq u} \ind_{s \in \cS} [f(s, \rho_{s^-} + \Delta \rho_s) - f(s^-, \rho_{s^-})]\, ,
\end{align}
where we used the fact that $\rho_{s^-} = \rho_s$, for any $s \in \cT \setminus \cS$.

Therefore,
\begin{align}\label{eq:first_int}
    &\mathop{\phantom{=}} \int_t^u \de^{-R_{s^-}} \, \dd f(s, \rho_s)
    = \int_t^u \de^{-R_{s^-}} \cL f(s, \rho_s) \, \dd s + \int_t^u \de^{-R_{s^-}} \sigma(s, \rho_s) \partial_x f(s, \rho_s) \, \dd W_s \notag
    \\
    &\qquad + \sum_{t < s \leq u} \ind_{s \in \cT \setminus \cS} \, \de^{-R_{s^-}} [f(s, \rho_s) - f(s^-, \rho_s)] + \sum_{t < s \leq u} \ind_{s \in \cS} \, \de^{-R_{s^-}} [f(s, \rho_{s^-} + \Delta \rho_s) - f(s^-, \rho_{s^-})] \, .
\end{align}

To compute the second integral in~\eqref{eq:product_formula}, we observe that
\begin{equation}\label{eq:exp_semimart}
    \de^{-R_u} = 1 - \int_t^u \de^{-R_s} \rho_s \, \dd s + \sum_{t < s \leq u} \ind_{s \in \cT} \left[\de^{-R_s} - \de^{-R_{s^-}}\right] \, ,
\end{equation}
and hence
\begin{equation}\label{eq:second_int}
    \int_t^u f(s^-, \rho_{s^-}) \, \dd\bigl(\de^{-R_s}\bigr) = -\int_t^u \de^{-R_s} \rho_s f(s, \rho_{s}) \, \dd s + \sum_{t < s \leq u} \ind_{s \in \cT} f(s^-, \rho_{s^-}) \left[\de^{-R_s} - \de^{-R_{s^-}}\right] \, .
\end{equation}

To compute the third integral in~\eqref{eq:product_formula}, we proceed as follows.
Denote $X_u = \de^{-R_u}$ and $Y_u = f(u, \rho_u)$, $u \in [t,T]$. Then, using~\citep[Theorem~4.52]{jacod2013:limit}, we deduce from~\eqref{eq:f_semimart} and~\eqref{eq:exp_semimart} that
\begin{equation*}
    \int_t^u \dd [\de^{-R_\cdot}, f(\cdot, \rho_\cdot)]_s = \int_t^u \dd [X,Y]_s = \int_t^u \dd \langle X^c, Y^c \rangle_s + \sum_{t < s \leq u} \Delta X_s \Delta Y_s = \sum_{t < s \leq u} \Delta X_s \Delta Y_s \, ,
\end{equation*}
since the continuous martingale part\footnote{Cf.~\citep[Proposition~4.27]{jacod2013:limit}.} $X^c$ of $X$ is zero. From~\eqref{eq:exp_semimart} we have that, for any $s \in (t,u]$,
\begin{equation*}
    \Delta X_s = \ind_{s \in \cT} \left[\de^{-R_s} - \de^{-R_{s^-}}\right],
\end{equation*}
while, from~\eqref{eq:f_semimart}, we obtain
\begin{equation*}
    \Delta Y_s = \ind_{s \in \cT \setminus \cS} [f(s, \rho_s) - f(s^-, \rho_s)] + \ind_{s \in \cS} [f(s, \rho_{s^-} + \Delta \rho_s) - f(s^-, \rho_{s^-})] \, .
\end{equation*}
Therefore we get
\begin{align}\label{eq:third_int}
    &\mathop{\phantom{=}}
    \int_t^u \dd [\de^{-R_\cdot}, f(\cdot, \rho_\cdot)]_s = \sum_{t < s \leq u} \Delta X_s \Delta Y_s \notag
    \\ 
    &= \sum_{t < s \leq u} \ind_{s \in \cT \setminus \cS} \left[\de^{-R_s} - \de^{-R_{s^-}}\right] [f(s, \rho_s) - f(s^-, \rho_s)] \notag
    \\ 
    &\quad
    + \sum_{t < s \leq u} \ind_{s \in \cS \cap \cT} \left[\de^{-R_s} - \de^{-R_{s^-}}\right] \left[f(s, \rho_{s^-} + \Delta \rho_s) - f(s^-, \rho_{s^-})\right] \, .
\end{align}

Summing~\eqref{eq:first_int}, \eqref{eq:second_int}, and~\eqref{eq:third_int} we get from~\eqref{eq:product_formula} that
\begin{align}\label{eq:semimart}
    &\mathop{\phantom{=}} \de^{-R_u} f(u, \rho_u) 
    = f(t, x) + \int_t^u \de^{-R_{s^-}} [\cL f(s, \rho_s) - \rho_s f(s, \rho_s)] \, \dd s + \int_t^u \de^{-R_{s^-}} \sigma(s, \rho_s) \partial_x f(s, \rho_s) \, \dd W_s \notag
    \\
    &+ \sum_{t < s \leq u} \ind_{s \in \cS} \de^{-R_{s^-}} \left[f(s, \rho_{s^-} + \Delta \rho_s) - f(s^-, \rho_{s^-})\right] + \sum_{t < s \leq u} \ind_{s \in \cT \setminus \cS} \de^{-R_{s^-}} [f(s, \rho_s) - f(s^-, \rho_s)] \notag
    \\
    &+ \sum_{t < s \leq u} \ind_{s \in \cT} f(s^-, \rho_{s^-}) \left[\de^{-R_s} - \de^{-R_{s^-}}\right] \notag
    + \sum_{t < s \leq u} \ind_{s \in \cT \setminus \cS} \left[\de^{-R_s} - \de^{-R_{s^-}}\right] [f(s, \rho_s) - f(s^-, \rho_s)] \notag
    \\
    &+ \sum_{t < s \leq u} \ind_{s \in \cS \cap \cT} \left[\de^{-R_s} - \de^{-R_{s^-}}\right] \left[f(s, \rho_{s^-} + \Delta \rho_s) - f(s^-, \rho_{s^-})\right]
    \, . 
\end{align}

Writing the set $\cS$ as the disjoint union of the sets $\cS \setminus \cT$ and $\cS \cap \cT$, and the set $\cT$ as the disjoint union of the sets $\cT \setminus \cS$ and $\cS \cap \cT$, after some simplifications we get that the terms in the last three lines of~\eqref{eq:semimart} sum up to
\begin{align*}
    &\mathop{\phantom{+}}\sum_{t < s \leq u} \ind_{s \in \cT \setminus \cS} \de^{-R_{s^-}} \left[\de^{-\rho_s} f(s, \rho_s) - f(s^-, \rho_s)\right]
    \\
    &+ \int_t^u \ind_{s \in \cS \setminus \cT} \int_\R \de^{-R_{s^-}} \left[f(s, \rho_{s^-} + z) - f(s^-, \rho_{s^-})\right] m^\cS(\dd s \, \dd z)
    \\
    &+ \int_t^u \ind_{s \in \cS \cap \cT} \de^{-R_{s^-}} \left[\de^{-(\rho_{s^-} + z)} f(s, \rho_{s^-} + z) - f(s^-, \rho_{s^-})\right] m^\cS(\dd s \, \dd z) \, ,
\end{align*}
where $m^\cS$ is the random measure introduced in~\eqref{eq:mS}. Compensating the last two terms above with the random measure $\mu^\cS$ given in~\eqref{eq:muS}, and substituting the result in equation~\eqref{eq:semimart}, we get
\begin{multline}\label{eq:semimart2}
    \de^{-R_u} f(u, \rho_u) 
    = f(t, x) + \int_t^u \de^{-R_{s^-}} \dd M^f_s 
    \\
    + \int_t^u \de^{-R_{s^-}} [\cL f(s, \rho_s) - \rho_s f(s, \rho_s)] \, \dd s + \sum_{t < s \leq u} \ind_{s \in \cT \setminus \cS} \, \de^{-R_{s^-}} \left[\de^{-\rho_s} f(s, \rho_s) - f(s^-, \rho_s)\right]
    \\
    + \int_t^u \ind_{s \in \cS \setminus \cT} \int_\R \de^{-R_{s^-}} \left[f(s, \rho_{s^-} + z) - f(s^-, \rho_{s^-})\right] \mu^\cS(\dd s \, \dd z)
    \\
    + \int_t^u \ind_{s \in \cS \cap \cT} \int_\R \de^{-R_{s^-}} \left[\de^{-(\rho_{s^-} + z)} f(s, \rho_{s^-} + z) - f(s^-, \rho_{s^-})\right] \mu^\cS(\dd s \, \dd z).
\end{multline}
By~\citep[4.34(b)]{jacod2013:limit}, the stochastic integral with respect to $M^f$ is a local martingale; also the discounted price process $\de^{-R_u}f(u,\rho_u)$, $u \in [t,T]$, is a martingale (hence also a local martingale). Therefore, the finite variation terms appearing in the last three lines of~\eqref{eq:semimart2} need to vanish. This, coupled with the assumption that $\rho_s$ has full support for any $s \in [t,T]$, and by the arbitrariness of $(t,x) \in [0,T)$, gives us PDEs~\eqref{eq:lastPDE} and~\eqref{finalSystem}, together with the terminal condition $f(T,x) = H(x)$, which is due to the definition of $f$, and the left-limit conditions, which arise from the fact that the last three terms in~\eqref{eq:semimart2} must be equal to zero.
\end{proof}

\begin{remark}
    If $r_K=T$, then PDE~\eqref{eq:lastPDE} is not to be considered. Indeed, in that case we have $f(r_K, x) = f(T,x) = H(x)$, and we proceed to solve the recursive system of PDEs~\eqref{finalSystem} beginning with the one on the domain $[r_{K-1}, r_K) \times \R$, with left limit condition
    \begin{equation*}
        f(r_K^-,x) = 
        \begin{dcases}
            \de^{-x} H(x) \, , & x \in \R, \, r_K \in \cT \setminus \cS\, , \\
            \int_\R H(x+z) \, Q_{m(K)}(\dd z) \, , & x \in \R, \, r_K \in \cS \setminus \cT \, , \\
            \int_\R \de^{-(x+z)} H(x+z) \, Q_{m(K)}(\dd z) \, , & x \in \R, \, r_K \in \cS \cap \cT \, .
        \end{dcases}
    \end{equation*}
\end{remark}

{
We conclude this section pointing out that if~\ref{hp:PDE:existuniqweak} in Remark~\ref{rem:weaker_assumpt} is assumed, instead of point~\ref{hp:PDE:existuniq} of Assumption~\ref{hp:PDE}, then the pricing function $f$ in~\eqref{eq:pricefunction} is defined in $[0,T]\times A$, and the PDEs in~\eqref{eq:lastPDE} and \eqref{finalSystem} are to be solved, respectively, in the domains $[r_K, T] \times A$, and $[r_{k-1}, r_k] \times A$, $k \in \{1, \dots, K\}$, where $A \subseteq \R$ is the state space of the RFR $\rho$ and is a closed interval with non-empty interior.
}

\section{Existence, uniqueness and regularity results}\label{App:ExUniq}
In the previous section, we showed that the price of a European-style derivative can be computed by solving a family of PDEs of the form
    \begin{equation}\label{eq:genericPDE}
    \begin{dcases}
        \cL f(t, x) - x f(t, x) = 0 \, , & (t,x) \in [\tau_0, \tau) \times \R \, , \\
        f(\tau,x) = \phi(x) \, , & x \in \R \, ,
    \end{dcases}
\end{equation}
where $\cL$ is the parabolic operator introduced in~\eqref{L}, $[\tau_0, \tau) \subset [0,T)$ is a given interval, and $\phi \colon \R \to \R$ is a given measurable function.
In this section, we rigorously analyze the well-posedness of this Cauchy problem and establish a Feynman-Kač representation for its solution. 
{Note that, for the reasons anticipated in Remark~\ref{rem:weaker_assumpt}, which we will also further detail below, we will not study the above PDE in domains of the type $[\tau_0,\tau] \times A$, where $A$ is as in~\ref{hp:PDE:existuniqweak} in Remark~\ref{rem:weaker_assumpt}.}

The existence, uniqueness and regularity results for the solution to the Cauchy problem for linear parabolic differential equations in unbounded domains are well understood when the coefficients are bounded, see for instance \cite{Friedman}, \cite{Lunardi}, where the problem is studied with analytical methods, and \cite{Freidlin} for a probabilistic approach. 
In financial problems, the coefficients and the final condition, and hence the solution, {typically} satisfy polynomial growth conditions, such as \( f(t,x) = O(1 + |x|^\eta)\) for some exponent \( \eta\geq 0 \).
In this case, the asymptotic behavior of the solution is determined by the final condition and by the behavior at infinity of the coefficients of the differential operator. 

In the case of a uniformly parabolic operator of the form (\ref{L}), the solution behavior will typically depend on the size of the PDE's coefficients. If $x \to +\infty$, then the solution may decay faster, as the potential term $x f(t,x)$ in (\ref{eq:genericPDE}) "pushes" the solution toward zero. Conversely, a negative potential (i.e. as $x \to -\infty$) might allow for growth or even cause blow-up in some cases. 
Usual assumptions are (locally) Lipschitz continuous coefficients, {polynomial} growth for the final condition and potential coefficient bounded from below, see \cite{Baldi}, \cite{Lunardi1998}, \cite{Rubio2010} and the references therein. 

{The cited literature shows that, even if the terminal payoff $H$ at 
$T$ in \eqref{eq:lastPDE} has polynomial growth, the solution $f(t,x)$
on $[r_K, T)$ may nevertheless exhibit exponential growth. Consequently, the terminal condition at $r_K$ in \eqref{finalSystem}, associated with the PDE in the preceding interval $[r_{K-1},r_K)$ and depending on the solution on $[r_K, T)$, may fail to have polynomial growth, and this phenomenon may propagate backward.
Moreover, the coefficient of the potential term is unbounded, so the cited results do not, in principle, apply.}

{Lemma \ref{LemmaExistence} below guarantees that, on each interval between two jumps, the solution of \eqref{eq:genericPDE} exhibits exponential growth even when the terminal condition itself has exponential growth and the potential term is $x f(t,x)$. To the best of our knowledge, this result has not been established elsewhere.}

We make the following assumptions:
\begin{assumption}\label{hp:FeynmanKac}
    \mbox{}
    \begin{enumerate}[label=\arabic*.]
        \item $\mu(t,x) = \alpha(t) + \beta(t) x$, where $\alpha,\beta \colon [\tau_0,\tau] \to \R$, are Lipschitz continuous functions on $[\tau_0,\tau]$; 
    \item\label{hp:FeynmanKac:sigma} $\sigma(t,x)$ is strictly positive, bounded and Lipschitz continuous on $[\tau_0,\tau]\times \mathbb{R}$, with $ \sigma^2(t,x)\geq \lambda_0 > 0 $;  
    \item\label{hp:FeynmanKac:final} the final condition $\phi$ is a continuous function and satisfies the exponential growth condition $|\phi(x)| \leq \bar C \exp(\bar c|x|)$, for some positive constants $\bar C$, $\bar c$.
    \end{enumerate}
\end{assumption}

\noindent The following lemma holds true
\begin{lemma}
    Under Assumption \ref{hp:FeynmanKac}, the PDE (\ref{eq:genericPDE}) admits a unique solution $f\in \mathcal{C}^{1,2}([\tau_0,\tau)\times \mathbb{R})$ in the class of functions such that $f(t,x)\leq \bar{K}\exp(k|x|)$, for some constants $\bar{K}, k>0$. 
\label{LemmaExistence}\end{lemma}
\begin{proof}
Defining the auxiliary function $v(t,x)=f(t,x)e^{c(t)x}$, where $c(t)$ is a $\mathcal{C}^{1}([\tau_0,\tau])$ function, and plugging it into the PDE, we get the transformed equation
$$\partial_t v(t,x) + \dfrac{\sigma^2(t,x)}{2} \partial_{x x}^2 v(t,x) + \left(\alpha(t) + \beta(t) x - c(t) \sigma^2(t,x)\right) \partial_{x} v(t,x) + $$
$$+\left( c^2(t)\dfrac{\sigma^2(t,x)}{2} -(c'(t)+c(t)\beta(t)+1)x -c(t)\alpha(t)\right)v(t,x)=0.$$
We can choose $c(t)$ such that
\begin{equation}\label{eq:c}
    c'(t)+c(t)\beta(t)+1=0, \quad t \in [\tau_0, \tau] \, .
\end{equation}
Note that the solution to the ODE~\eqref{eq:c} is bounded and Lipschitz continuous, for any given terminal condition $c(\tau) = c_\tau$\footnote{The solution to (\ref{eq:c}) is
\begin{equation*}
    c(t)=\de^{\int_{t}^{\tau}\beta(u) \dd u}\left(c_\tau+\int_{t}^{\tau} \de^{-\int_{s}^{\tau}\beta(u)\dd u} \dd s\right) \, , \quad t \in [\tau_0, \tau] \, .
\end{equation*}}. In particular, we set $c_\tau = 0$, so that $v(\tau,\cdot)=f(\tau,\cdot)$.
The transformed PDE reads as
\begin{multline}
    \partial_t v(t,x) + \dfrac{\sigma^2(t,x)}{2} \partial_{x x}^2 v(t,x) + \left(\alpha(t) + \beta(t) x - c(t) \sigma^2(t,x)\right) \partial_{x} v(t,x) 
    \\
    - \left( c(t)\alpha(t) - c^2(t)\dfrac{\sigma^2(t,x)}{2} \right)v(t,x)=0,
\label{TransfPDE}
\end{multline}
with final condition $v(\tau,x)=\phi(x)$, where the potential term coefficient is now bounded from below and also locally Lipschitz continuous, as it is a combination of Lipschitz functions.
Therefore, we can apply Theorem 10.6 in \cite{Baldi} which ensures the existence and uniqueness of the solution $v\in \mathcal{C}^{1,2}([\tau_0,\tau)\times \mathbb{R})$,
satisfying the growth condition $|v(t,x)|\leq \bar K e^{\bar k|x|}$, for some positive constants $\bar K$, $\bar k$. Substituting into $f(t,x)=v(t,x)e^{-c(t)x}$, we get the desired result.
\end{proof}

{The next theorem provides us with a Feynman-Kač representation formula for the solution $f$ of PDE~\eqref{eq:genericPDE}.
To the best of our knowledge, this formula is not available in the literature, where PDEs of the following form are considered
\begin{equation*}
    \begin{dcases}
        \cL f(t, x) - c(x) f(t, x) = 0 \, , & (t,x) \in [0,T) \times \R \, , \\
        f(T,x) = \phi(x) \, , & x \in \R \, ,
    \end{dcases}
\end{equation*}
with the potential coefficient $c$ bounded from below (see, e.g.,~\citep{Baldi,karatzas:brown,palczewski2025partial}). In our case, $c(x) = x$, thus this coefficient is unbounded in $\R$, which prevents us from applying standard results. However, under the setting established by Assumption~\ref{hp:FeynmanKac}, we are able to perform a probability measure change, which allows us to recover the Feynman-Kač representation formula.
}

\begin{theorem}\label{th:FK}
Under the Assumption \ref{hp:FeynmanKac}, the solution to equation (\ref{eq:genericPDE}) is given by  
$$f(t,x) = \bbE\left[\phi(X_\tau^{t,x})\de^{-\int_t^\tau X_s^{t,x} \dd s}\right]   \qquad \qquad \qquad \mbox{for } (t,x) \in [\tau_0, \tau] \times \R,$$
\label{FKformula}
where $X^{t,x}$ satisfies the equation \begin{equation*}
    \begin{dcases}
    \dd X_s=\left(\alpha(s) + \beta(s) X_s\right) \dd s + \sigma(s,X_s) \dd W_s,  & \mbox{for } s \in (t,\tau] \\
    X_t = x. \, 
    \end{dcases}
\end{equation*}
\end{theorem}
\begin{proof}
By the Feynman-Kač representation of the solution to the transformed PDE (\ref{TransfPDE}) satisfied by $v$ we have  

$$v(t,x)= \mathbb{E}\left[\phi(\tilde X_\tau^{t,x})\de^{-\int_t^{\tau} \Gamma(s, \tilde X_s^{t,x}) ds}\right] \ ,$$
where process $\tilde X$ solves the SDE 
\begin{equation}
    \left\{\begin{array}{ll}
    \dd \tilde X_s=\left(\alpha(s) + \beta(s) \tilde X_s - \sigma^2(s,\tilde X_s)  c(s)\right) \dd s + \sigma(s,\tilde X_s) \dd W_s,  & \mbox{for } s \in (t,\tau], \\
   \tilde X_t = x,  & 
\end{array}\right.
\label{SDE_G}\end{equation}
and $\Gamma(t,x) \coloneqq \alpha(t) c(t) - \frac 12 \sigma^2(t,x)c^2(t)$. 
Note that $B_s = W_s - \int_t^s \sigma(u, \tilde X_u^{t,x})  c(u) \, \dd u$ is a standard Brownian motion under $\bbQ^*$, whose density with respect to $\bbQ$ is  
\begin{equation*}
    \dfrac{\dd \bbQ^*}{\dd \bbQ} = \exp\left\{\int_t^\tau \sigma(u, \tilde X_u^{t, x})  c(u) \, \dd W_u - \frac 12 \int_t^\tau  \sigma^2(u, \tilde X_u^{t, x})  c^2(u) \, \dd u \,\right\} \, .
\end{equation*}
{Indeed, since $\sigma$ is bounded by Assumption~\ref{hp:FeynmanKac}.\ref{hp:FeynmanKac:sigma}, and $c$ is bounded by construction (see the proof of Lemma \ref{LemmaExistence}), Novikov's condition is satisfied. Therefore the density $ {\dd \bbQ^*}/{\dd \bbQ}$ is a true martingale and Girsanov's theorem applies.} The $\bbQ^*$-dynamics of $\tilde{X}$ are
\begin{equation*}
    \begin{dcases}
    \dd \tilde X_s=\left(\alpha(s) + \beta(s) \tilde X_s\right) \dd s + \sigma(s,\tilde X_s) \dd B_s,  & \mbox{for } s \in (t,\tau], \\
    \tilde X_t = x \, ,
    \end{dcases}
\end{equation*}
i.e., the law of $\tilde X$ under $\bbQ^*$ coincides with the law of $X$ under $\bbQ$. Then, we get
\begin{align*}
    v(t,x) 
    &= \mathbb{E}\left[\phi(\tilde X_\tau^{t,x})\de^{-\int_t^{\tau} \Gamma(s, \tilde X_s^{t,x}) \, \dd s}\right] 
    \\
    &= \bbE^{\bbQ^*}\left[\phi(\tilde X_\tau^{t,x})\de^{-\int_t^{\tau} \Gamma(s, \tilde X_s^{t,x}) \, \dd s} \dfrac{\dd \bbQ}{\dd \bbQ^*}\right] 
    \\
    &= \bbE^{\bbQ^*}\left[\phi(\tilde X_\tau^{t,x})\de^{-\int_t^{\tau} \alpha(s)  c(s) \, \dd s - \int_t^\tau \sigma(s, \tilde X_s^{t,x})  c(s) \, \dd B_s}\right] 
    \\
    &= \bbE\left[\phi(X_\tau^{t,x})\de^{-\int_t^{\tau} \alpha(s)  c(s) \, \dd s - \int_t^\tau \sigma(s, X_s^{t,x})  c(s) \, \dd W_s}\right] \, .
\end{align*}

Using the Ito's product formula we have that
\begin{equation*}
    \alpha(s)  c(s) \, \dd s + \sigma(s, X_s^{t,x})  c(s) \, \dd W_s = \dd(X_s^{t,x}  c(s)) + X_s^{t,x} \, \dd s \, ,
\end{equation*}
whence
\begin{equation*}
    \int_t^\tau \alpha(s)  c(s) \, \dd s + \int_t^\tau \sigma(s, X_s^{t,x})  c(s) \, \dd W_s = \underbrace{X_\tau^{t,x}  c(\tau)}_{=0} - x  c(t) + \int_t^\tau X_s^{t,x} \, \dd s \, .
\end{equation*}
Therefore,
\begin{align*}
    v(t,x) 
    &= \bbE\left[\phi(X_\tau^{t,x})\de^{-\int_t^{\tau} \alpha(s)  c(s) \, \dd s - \int_t^\tau \sigma(s, X_s^{t,x})  c(s) \, \dd W_s}\right]
    \\
    &= \bbE\left[\phi(X_\tau^{t,x})\de^{-\int_t^\tau X_s^{t,x} \, \dd s} \de^{x  c(t)}\right] \, ,
\end{align*}
and hence
\begin{equation*}
    f(t,x) = \de^{-x  c(t)}  v(t,x) = \bbE\left[\phi(X_\tau^{t,x})\de^{-\int_t^\tau X_s^{t,x} \, \dd s}\right] \, . \qedhere
\end{equation*}
\end{proof}
\bigskip
{
It is worth emphasizing that thanks to the probabilistic representation of solutions to the family of PDEs.~\eqref{eq:genericPDE}, we are able to connect the results of this section to the existence and uniqueness of a pricing function for interest-rate derivatives, with suitably integrable payoffs, as seen in Section~\ref{PDEformulation}. Indeed, to find the pricing function of an interest rate derivative with contract function $H$, satisfying Assumption~\ref{hp:PDE}.\ref{hp:PDE:H} and Assumption~\ref{hp:FeynmanKac}.\ref{hp:FeynmanKac:final}, we can solve a sequence of PDEs of the form~\eqref{eq:genericPDE} according to the following scheme:
\begin{itemize}
    \item Solve PDE~\eqref{eq:genericPDE} with $\phi = H$ and $[\tau_0, \tau) = [r_K, T)$, where $r_K \in \cR$ is the last relevant date (see Section~\ref{PDEformulation}), to get the solution $f^K$;
    \item Solve PDE~\eqref{eq:genericPDE} on $[\tau_0, \tau) = [r_{K-1}, r_K)$ and with
    \begin{equation*}
        \phi(x) = 
        \begin{dcases}
            \de^{-x} f^K(r_K, x) \, , &\text{if } r_K \in \cT \setminus \cS\, , \\
            \int_\R f^K(r_K, x+z) \, Q_{m(K)}(\dd z) \, , &\text{if } r_K \in \cS \setminus \cT \, , \\
            \int_\R \de^{-(x+z)} f^K(r_K, x+z) \, Q_{m(K)}(\dd z) \, , &\text{if } r_K \in \cS \cap \cT \, ,
        \end{dcases}
    \end{equation*}
    to get the solution $f^{K-1}$.
    \item Proceed backwards in this way solving PDE~\eqref{eq:genericPDE} in each interval $[\tau_0, \tau) = [r_{k-1}, r_k)$, $k \in \{1, \dots, K-1\}$ and with $\phi$ similar as above (use $f^k$ and $r^k$ whenever $f^K$ and $r^K$ appear), to determine the sequence of solutions $f^{K-2}, \dots, f^0$.
    \item Piece together solutions $f^0, \dots, f^K$ in each interval $[0, r_1)$, \dots, $[r_{K-1}, r_K)$, $[r_K, T]$, to get the pricing function $f$ for contingent claim $H$, which verifies the risk-neutral valuation formula~\eqref{eq:pricefunction}.
\end{itemize}%
}

We conclude this section with a remark. In the proof of Lemma \ref{LemmaExistence} and Theorem \ref{FKformula}, we choose, for simplicity, $c_\tau = 0$. Clearly, the choice of a different terminal condition for the ODE satisfied by $c$ yields a different transformation of the original function $f$, which however satisfies the same PDE written above for $v$, where the only thing changing is the actual expression of the function $c$. In particular, this provides us with several equivalent ways of transforming the original function $f$. 
For example, in the Vasicek model where $\alpha$, $\beta$ and $\sigma$ are assumed to be constant, it is convenient to choose $c_\tau= - \frac{1}{\beta}$ so that $c$ is constant, and consider the transformed function $v(t,x)=f(t,x)e^{-x/\beta}$. 
Therefore, the transformed PDE becomes
\begin{equation}\label{eq:Vas_cambio}
\partial_t v(t,x) + \dfrac{\sigma^2}{2} \partial_{x x}^2 v(t,x) + \left(\alpha(t) + \beta x + \dfrac{\sigma^2}{\beta}\right) \partial_{x} v(t,x) + \left( \dfrac{\sigma^2}{2\beta^2} + \dfrac{\alpha(t)}{\beta}\right)v(t,x)=0.    
\end{equation}
Equation~\eqref{eq:Vas_cambio} will be exploited in Appendix~\ref{Localization}.

\section{Affine models} \label{affinesection}

In this section we apply our results to the case of an affine term structure, where the dynamics of the RFR are described by the following equation
\begin{equation}
	\dd\rho_t= \left (\alpha(t) + \beta(t) \rho_t \right) \dd t +  \sqrt{\gamma(t) + \delta(t) \rho_t \, } \, \dd W_t+ \dd J_t,
\label{shortrateaff}\end{equation}
where $\alpha, \beta, \gamma, \delta: \mathbb{R^+} \to \mathbb{R}$, with $ {\gamma(t)}, \delta (t) \ge 0$ for all $t$,  are assumed to be sufficiently regular to guarantee the existence and uniqueness of a strong solution  {between jump times}, as well as the non-negativity of the diffusion coefficient. 

We remark that the conditions outlined in the previous section are sufficient, but not necessary, to guarantee the existence of a strong solution.
For example, the CIR process does not satisfy the Lipschitz continuity condition for the diffusion coefficient. Nevertheless, it admits a well-defined strong solution between jump times, provided that the Feller condition is satisfied.     
In the CIR model, the process is typically used to model quantities that are required to remain non-negative, such as interest rates  and volatility. The Feller condition ensures that the process remains strictly positive, but it does not guarantee that it stays bounded away from zero. Therefore, if jumps are allowed to take even slightly negative values, it may happen (with a probability that is possibly  very small)  that, following a jump, the interest rate process may fall below zero. In such a case, the initial condition of the subsequent stochastic differential equation becomes ill-posed, as the square-root diffusion term is not defined for negative values. This implies that, in the CIR model, or more generally, in affine models with square-root diffusion terms, introducing stochastic discontinuities requires restricting jumps to non-negative values in order to preserve the \mbox{well-posedness} of the process.
It is clear that this assumption is not particularly realistic. In general, one can introduce additional mechanisms, such as reflecting barriers or absorbing boundaries, to prevent the process from becoming negative. In this case, the model remains feasible even in the presence of negative jumps, but it no longer belongs to the class of affine processes.  
{Very recently  \cite{fontanaPS} proposed an extension of the CIR model that incorporates stochastic discontinuities, which preserves affinity and non-negativity while allowing upward and downward jumps: in their model, the distribution of jump sizes depends explicitly on the pre-jump state of the process, satisfies some suitable conditions, and jumps sizes may exhibit autocorrelation.  
To keep the analysis simple and to mantain the independence of jumps}, the results presented in this  section implicitly assume that, whenever the function $\delta(t)$ is not identically zero, the jump sizes are non-negative.

Affine term {structures} with stochastic discontinuities have been studied by \cite{KellerRessel2019}  and  \cite{fontanaetal2024:termstruct}. In particular, in \cite{fontanaetal2024:termstruct} the authors compute the conditional characteristic function and employ it to derive closed-form expressions for the prices of derivatives, such as { ZCB}s and forward-looking caplets within an extended \mbox{Hull-White} model with Gaussian-sized jumps. A similar methodology is adopted in \cite{economies12030073},  under the assumption that jump sizes follow a discrete distribution governed by a modified Skellam law.

As an alternative, we illustrate how the classical approach, based on solutions of PDEs (see for instance \cite{bjork:arbitragetheory}, Chap. 21), can be adapted to this framework  and exploited to find solutions in closed-form. The proofs of all results presented in this section are gathered in Appendix \ref{proofs}.
For the sake of simplicity, we assume that  $\cS \cap \cT = \emptyset.$ 
System \eqref{finalSystem} then can be rewritten as: 
 
\begin{small}
\begin{equation}
\label{affinePE}  
\begin{dcases}  {\partial}_t f (t, x) +  \left (\alpha(t) + \beta(t) x \right)  {\partial}_x f (t, x) + \frac{1}{2}  \left (\gamma(t) + \delta(t) x \right) \, {\partial_{xx}^2}  f   (t, x) =  f (t,x) x,  &(t,x) \in [r_{k-1}, r_k) \times \R, \\
		f(T, x) = H(x), & x \in \mathbb{R}, \\
        f(r_k^-, x) = \de^{-x} f(r_k, x), & r_k \in \cT, \ 
        x \in \mathbb{R},\\ 
        f(r_k^-, x) = 	    \mathbb{E}[f(r_k, x  + \xi_j)],    & r_k=s_j \in \cS, x \in \mathbb{R}.
\end{dcases}
\end{equation}
\end{small}

{It is worth recalling that if~\ref{hp:PDE:existuniqweak} in Remark~\ref{rem:weaker_assumpt} is assumed, instead of point~\ref{hp:PDE:existuniq} of Assumption~\ref{hp:PDE}, then the PDEs in the above system and PDE~\eqref{eq:lastPDE} are to be solved in the domains $[r_{k-1}, r_k] \times A$, $k \in \{1, \dots, K\}$, and $[r_K, T]$, where $A \subseteq \R$ is the state space of the RFR $\rho$ and is assumed to be a closed interval with non-empty interior.}

\subsection{Price of a zero-coupon bond} 
Let {$P_T$ denote the pricing function of a ZCB maturing at time $T$ (also called $T$-bond), i.e., $P_T(t, x)$ is the price at time $t$ of a $T$-bond, given that $\rho_t = x$. We recall that this function is defined by~\eqref{eq:pricefunction}, since any $T$-bond can be seen as a contingent claim written on the RFR $\rho$ with contract function $H(x) \equiv 1$.}
Given the affine term structure, the price of a { ZCB} is expected to exhibit an exponential-affine dependence on the short-rate. This is formally established in the following proposition, where we assume for simplicity that $\max(t_N, s_M) < T \eqqcolon t_{N+1} \eqqcolon s_{M+1}$. 

\begin{proposition}\label{propZCB}  The {pricing function  $P_T$ of a ZCB maturing at $T$} is given by 
    \begin{equation} \label{ZCBpriceGeneral} P_T(t,x) = \exp(-a(t,T) - x b(t,T)),
    \end{equation} 
where the  deterministic functions $a$ and $b$ are determined as follows:
\begin{enumerate}[(i)]
    \item $ \displaystyle \label{affineb} \displaystyle b(t,T) =  \sum_{n = 0}^{N} b_n(t) \ind_{[t_{n}, t_{n+1}) } (t) $,  
where $b_n$ denotes the  solution, on the interval $[t_{n}, t_{n+1})$, of the Riccati equation 
\begin{equation} \label{Riccati}
 b_n'(t) + \beta (t) b_n(t) - \frac{1}{2} \delta(t)  b_n^2 (t)  + 1  = 0    
\end{equation}
with terminal condition $b_n(t_{n+1}^-) = b_{n+1}(t_{n+1}) + 1$ for $n=0, \dots, N-1$, and final value $b_{N}(t_{N+1}) = b(T,T)=0$;
 
\bigskip 
 
\item \begin{minipage}{16.5cm}
\begin{equation} \label{affinea}
    \displaystyle \displaystyle \hspace{-3.2cm}a(t,T) =  \int_t^T \, \left[ \alpha (u) b(u,T) - \frac{1}{2} \gamma(u) b^2(u,T) \right] \dd u  - \sum_{j=1}^M \log \mathbb{E} \left[\de^{-\xi_j  b(s_j, T) }   \right] \ind_{[0, \, s_{j}) } (t). \hspace{0.cm} 
\end{equation}
\end{minipage}

\end{enumerate}

\end{proposition}

\bigskip 
Note that the coefficient $b = -\partial_x P_T/P_T$ which represents the duration (sensitivity) of the { ZCB} with respect to the instantaneous short-rate, is not affected by jumps in the dynamics of the interest rate,  but possibly only by discontinuities in the num\'eraire.

The following result establishes a relationship between the price of a { ZCB} in the current framework (with stochastic discontinuities and jumps in the num\'eraire)  and its counterpart in the continuous (\mbox{no-jump}) setting. The proof is omitted since it is based on standard computation.

\begin{corollary} \label{Pcont} Let $P^c_T(t,x) =\exp(-A(t,T) - x B(t,T))$ be the {pricing function} of a { ZCB}\footnote{Cf.~\citep[Proposition~22.2]{bjork:arbitragetheory}.} maturing at time $T$, assuming absence of jumps in both the dynamics of the short-rate and in the num\'eraire. Let $a(t,T)$, $b(t,T)$ be the functions defined in \eqref{ZCBpriceGeneral}.
Then, the following relations hold: 
 
\begin{enumerate}[(i)]
\item\label{Pcont:b} $\displaystyle b(t,T) =  B(t,T) + \sum_{n=1}^N B_n(t)\ind_{[0, \, t_n) } (t),$ where $B_n(t)$  denotes the solution, on the interval $[0, t_n)$, of  the Riccati equation
\begin{equation} \label{affineB}  B_n'(t) + \left[\beta (t) -\delta(t) {\mathfrak b}_{n}(t)\right] B_n(t) - \frac{1}{2} \delta(t) B_n^2 (t) = 0
\end{equation}
with terminal condition $B_n(t_{n}^-) = 1$, where 
${\mathfrak b}_n(t) = B(t,T) + \sum_{k = n+1}^N B_k(t)$ for $n=1, \ldots, N-1$,  ${\mathfrak b}_{N}(t) = B(t,T)$; 

\bigskip 
\item\label{Pcont:a}  $\displaystyle a(t,T) =  A(t,T) + \sum_{n=1}^N A_n(t) \ind_{[0, \, t_n) } (t) + \sum_{n =1}^{N-1} \sum_{l=n+1}^N C_{nl}(t) \ind_{[0, \, t_n) } (t) -  \sum_{j =1}^M \log \mathbb{E}[\de^{-\xi_j  b(s_j,T)}] \ind_{[0, \, s_j) } (t), $ 
where 
\begin{eqnarray*} & & A_n(t) = \int_t^{t_n} \left[ (\alpha (u) - \gamma(u) B(u,T)) \,   B_n (u) - \frac{1}{2}   \gamma(u) B^2_n(u) \right] \dd u,\\  
& &   C_{nl}(t) = - \int_t^{t_n}  \gamma(u) B_n(u) B_l(u) \, \dd u. 
\end{eqnarray*}
\end{enumerate}
\end{corollary} 

\bigskip 
\begin{remark} The terms $ \log \mathbb{E}[\de^{-\xi_j  b(s_j,T)}]$ depend solely on the distribution of $\xi_j$, and can be explicitly computed once this distribution is specified. In Section \ref{sec:simulation}, we evaluate these expressions in the two cases where $\xi_j$
follows either a normal distribution or a discrete distribution with two possible outcomes.
\end{remark}

\bigskip 
 
As an example we consider the Vasicek model. In \eqref{affinePE}, let $\alpha(t) = \alpha$, $\beta(t) = \beta$, $\gamma(t) = \sigma^2$ and   $\delta (t) = 0$, with $\alpha \in \mathbb{R}, \beta \in \mathbb{R^-}, \sigma \in \mathbb{R^+}.$ 
In this case, the model reduces to a Vasicek short-rate model with discontinuities. From standard results (see for instance \cite{bjork:arbitragetheory}, Proposition 22.3) we derive$$B(t,T) =\frac{\de^{\beta(T-t)}-1}{\beta}, \qquad   A(t,T)= \frac{\alpha}{\beta} \left[B(T-t) - (T-t)\right] -  \frac{\sigma^2}{2\beta^2}  \left[\frac{\beta}{2} B^2(T-t)  - B(T-t) + (T-t)\right].$$
Equation~\eqref{affineB} in Corollary~\ref{Pcont}~\ref{Pcont:b} becomes
$ B'_n(t) + \beta B_n(t) = 0
$, which, together with the terminal condition $B_n(t_n^-) = 1$ yields $B_n(t) = e^{\beta(t_n -t)}.$   
So we obtain
\begin{eqnarray*}b(t,T) &=& \frac{\de^{\beta(T-t)}-1}{\beta} + \sum_{n=1}^{N} \de^{\beta(t_n -t)} \ind_{[0, \, t_n) } (t).  
\end{eqnarray*}

\bigskip  

\noindent Applying Corollary~\ref{Pcont}~\ref{Pcont:a}, we  find $a(t,T)$  as follows: 

\begin{eqnarray*} 
a(t,T) & =& A(t,T) +  \left(\alpha + \frac{\sigma^2}{\beta} \right) \sum_{n=1}^N  B(t,t_n ) \ind_{[0, \, t_n) } (t)  \\ &-& \frac{\sigma^2}{2}   \, \sum_{n=1}^N   \left(\frac{\de^{\beta(T+t_n -2t)} - \de^{\beta(T-t_n)}}{\beta^2}+\frac{\de^{2\beta(t_n -t)} - 1}{2\beta}\right) \ind_{[0, \, t_n) } (t) \\& - &    \sigma^2 \sum_{n =1}^{N-1} \sum_{l=n+1}^N  \left( \frac{\de^{\beta(t_n + t_l -2t)}-\de^{\beta(t_l -t_n)  }}{2\beta} \right) \ind_{[0, \, t_n) } (t)  -\sum_{j =1}^M \log \mathbb{E}[\de^{-\xi_j  b(s_j,T)}] \ind_{[0, \, s_j) } (t).    
\end{eqnarray*}
   
\bigskip
\subsection{Price of a call option in the extended Vasicek model with gaussian jumps}

In this section, we derive the price of a call option under the assumption that the short-rate follows Vasicek-type dynamics with  stochastic discontinuities,  and normally distributed jump sizes, in order to compare it with the numerical approximations presented in the following sections.  For simplicity, we consider the case of constant coefficients; however, the result can be readily extended to accommodate time-dependent coefficients (Hull-White model). Additionally, we note that a similar approach can be applied to determine the price of a forward-looking caplet, which was calculated using a different method in \cite{fontanaetal2024:termstruct}. 

Let $f(t, \rho_t)$  denote the price at time $t$ of a call option  on a ZCB with bond maturity $S$, option expiration date $T<S$, and  strike price $\widehat{K}$. The payoff at time $T$ is given by 
$$H(\rho_T) = (P_S(T,\rho_T)- \widehat{K})^+.$$
A closed-form expression for $f(t, x)$ is provided in the next proposition, whose proof can be found in Appendix \ref{proofcall}.  Note that for the purposes of our calculation, it makes sense to consider the subintervals determined by the jump times or roll-over dates only up to time $T$. Any jumps occurring between 
$T$ and $S$ are implicitly incorporated into   $P_S$. Therefore we set, as in the previous subsections,  $t_{N+1}=s_{M+1}=r_{K+1}=T$.
\begin{proposition}\label{affinecallprop}   
Assume each jump size $\xi_j$ is normally distributed with mean $m_j$ and variance $\gamma^2_j$ for all $j=1, \ldots, M.$ For all $t \le T$, the {pricing function $f$} of a call option on a ZCB maturing at $S$, with option expiration date $T<S$, and  strike price $\widehat{K}$, is 
\begin{equation}\label{affinecallpricefinal}
	f(t,x) =  P_S(t, x) \Phi(d_1(t,x)) - \widehat{K} P_T(t, x) \Phi(d_2(t,x)), 
\end{equation}	 
 
\noindent where $\Phi$ denotes the cumulative distribution function of the standard Gaussian distribution,  
\begin{eqnarray} \label{affinecallfinal}
 	d_1 (t,x ) = \frac{1}{\sigma_c(t)} \log \left(\frac{P_S(t, x)}{P_T(t,x) \widehat{K} }\right) + \frac{\sigma_c(t)}{2}, \qquad 
	 d_2  (t,x) =  d_1  (t,x) - \sigma_c(t),   \end{eqnarray} 
and $\sigma_c$ is given by 
\begin{equation}  \label{sigmacallfinal}
	\sigma_c (t) = b(T,S)  \left(  \sqrt{\frac{\sigma^2}{2\beta} \left(\de^{2\beta (T-t)}-1 \right)  +  \sum_{j =1}^M \gamma_j^2 \de^{2\beta (T-s_j)}   \ind_{[0, \, s_j) } (t)
}	\right) \, .
     \end{equation}
\end{proposition}

\bigskip 
\begin{remark}  If the jumps are not normally distributed, obtaining a closed-form expression for the option price becomes challenging, as the RFR no longer follows a normal distribution once a non-normal jump is introduced. In such cases, numerical approximation becomes essential for computing the price. 
\end{remark}
 
\bigskip 

\section{Numerical results}\label{numerical_section}

In this section, we test our theoretical results in several numerical simulations. The aim is twofold: firstly, we propose a semi-analytic method and a finite difference method to solve the pricing PDE and test their accuracy for claims with closed-form solution; in doing so, as a second outcome, the usefulness, ease of use, and versatility of the PDE approach and of the numerical approximation methods will be highlighted in comparison to the use of closed-form formulas.

The numerical results are obtained mainly for the Vasicek model with stochastic discontinuities, because in this case some closed-form formulas are available and we derived the Green's function explicitly. However, we emphasize the potentiality of the PDE approach in more general cases, namely for more general contingent claims and interest rate models.

Also in this section, for the sake of simplicity, we assume $\cS \cap \cT = \emptyset$, that is, there are no common jump times between the RFR and the numéraire.
\bigskip

\subsection{Distribution of the risk-free rate jumps}\label{sec:simulation}

We considered two different distributions of the risk-free rate jumps: (i) Gaussian $\cN(m_j,\gamma_j^2)$ 
and (ii) a discrete distribution where $\xi_j=m_j$ with probability $p_j$ and $\xi_j=-m_j$ with probability $1-p_j$.  

Figure \ref{Trajectory} shows two possible trajectories of the Hull-White spot interest rates with stochastic discontinuities, showing that our model allows reproducing the main stylized facts of the overnight rates. 
The first panel is obtained by setting $\delta(t)=0$, $\sigma(t)=\sqrt{\gamma(t)} =0.2-0.1/(1+\exp(-2000(t-0.4)))$, $\alpha(t) =0.0075/(1+\exp(-2000(t-0.4)))$, $\beta(t) =-0.3-99.7/(1+\exp(-2000(t-0.4)))$ in (\ref{shortrateaff});
the jumps $\xi_j$ are normally distributed with $m_j=0.05$ and $\gamma_j=0.1$ and occur at $s_j=0.3, 0.5, 0.6, 0.8, 0.9$.
The second panel is obtained by setting 
$\sigma(t) =0.05+0.05/(1+\exp(-2000(t-0.4)))$, $\alpha(t) =0.1-0.0925/(1+\exp(-2000(t-0.4)))$, $\beta(t) =-0.3-99.7/(1+\exp(-2000(t-0.4)))$
and considering discretely distributed jumps with $m_j=0.04$ and $p_j=0.9$ occurring at $s_j=0.3, 0.5, 0.6, 0.8, 0.9$.
In both cases, due to the values of the SDE coefficients, the first jump produces a spike (small diffusive component and very high mean reversion speed) while the following jumps correspond to structural changes of the spot rates to a new level (slower mean reversion).
\begin{figure}[H]
\centering
    \includegraphics[width=0.95\textwidth]{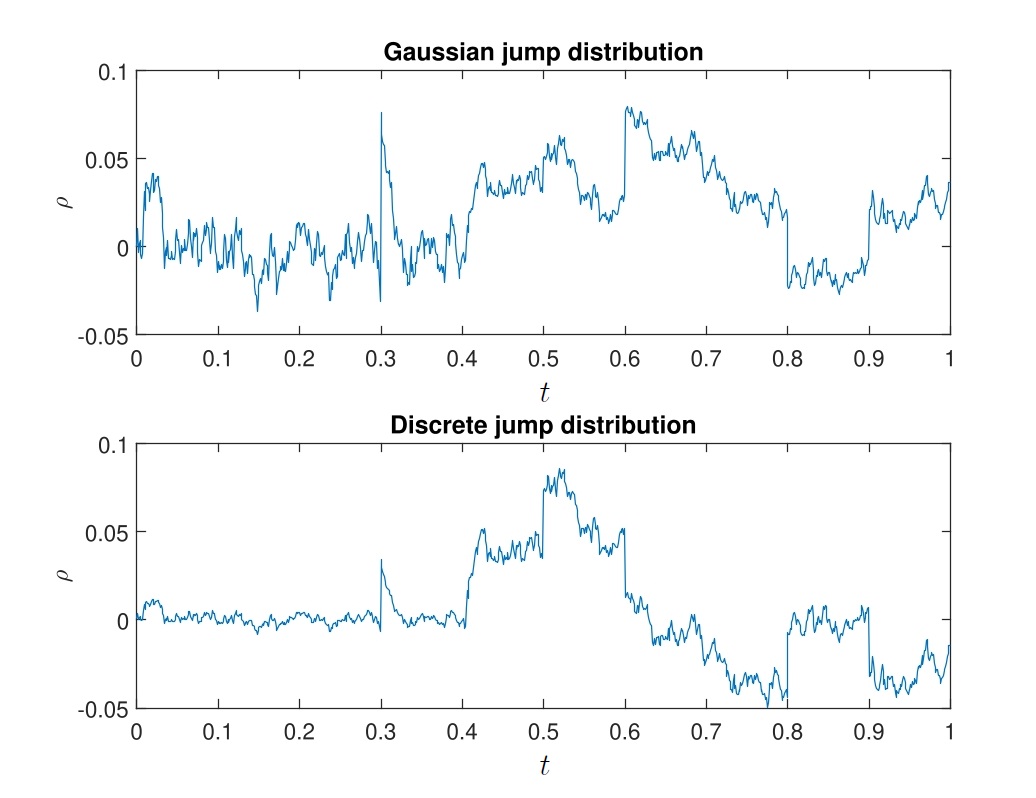}
	\caption{Hull-White spot interest rates with stochastic discontinuities at $s_j=0.3, 0.5, 0.6, 0.8, 0.9$.}\label{Trajectory}
\end{figure}

We compute the term
$$\log \mathbb{E} \left[\de^{-\xi_j  b(s_j, T) }   \right]$$
in the closed-form formula (\ref{affinea}). We get the following
\begin{enumerate}
    \item[{(i)}] $\displaystyle -m_jb(s_j,T)+\frac{\gamma_j^2}{2}b^2(s_j,T)$ for Gaussian jumps;
    \item[{(ii)}] $\log (\de^{-m_j b(s_j,T)}p_j+\de^{m_j b(s_j,T)}(1-p_j))$ for discrete jumps. 
\end{enumerate}

\subsection{Semi-analytic method}

We now describe the semi-analytic approach. We have proved in Section \ref{PDEformulation} that the price of any contingent claim at any time $t$ and for any interest rate $x$ can be obtained by solving a set of linear backward parabolic PDEs on each period $[r_k,r_{k+1}]$, for $k=0,1,\ldots,K$, with $r_{K+1}=T$.
Furthermore, applying Theorem~\ref{th:FK} in each interval $[r_k,r_{k+1})$, the solution to each of these PDEs uniquely exists and can be expressed by a Feynman-Kač formula 
\begin{equation}
    f(t,x)=\bbE\left[g(\rho_{r_{k+1}}^{t,x})\de^{-\int_t^{r_{k+1}} \rho_s^{t,x} \dd s}\right] =\int_{\mathbb{R}} G(t,r_{k+1};x,\xi) g(\xi) \, \dd\xi,
\label{solG}\end{equation}
where $g(\xi)=f(r_{k+1}^-,\xi)$ is the terminal condition and $G$ is the so-called \textit{fundamental solution} or \textit{Green's function} associated to the parabolic operator
$$\cL f(t, x) - x f(t, x).$$
For the reader's convenience, the analytic expression of the Green's function for the Vasicek model is provided in Appendix \ref{AppGreen}. 

Based on these considerations, we can build our semi-analytic method that, starting from maturity $T$, proceeds backward and computes the contingent claim price at the beginning of each period $[r_k,r_{k+1})$ by equation (\ref{solG}), using as final condition at $r_{k+1}$
\begin{equation}
g(\xi)=f(r_{k+1}^-,\xi)=\left\{
\begin{array}{ll}
  H(\xi),   &  \textrm{for} \ r_{k+1} = T,\\
  \int_{\mathbb{R}}f(r_{k+1},\xi+z) \, {Q}_{m(k+1)}(\dd z), & \textrm{for} \ r_{k+1} \ \in {\cal S},\\
  \de^{-\xi} f(r_{k+1},\xi),   & \textrm{for} \ r_{k+1} \ \in {\cal T}.
\end{array} \right.
\label{jump_condition}\end{equation}
In the following box, we summarize the steps of our algorithm.\\
\begin{algorithm}
\caption{\emph{Semi-analytic algorithm}}\label{AlgGreen}
\vspace{-0.2cm}{\begin{algorithmic}
    \STATE Starting from the payoff condition in $T$, {\bf compute} the solution in $[r_{K},T) \times \mathbb{R}$ using (\ref{solG}).
\FOR {$k=K,K-1,\ldots 2,1$}
\STATE {\bf Compute} the jump condition \eqref{jump_condition} at $\{r_{k}^-\} \times \mathbb{R}$;
\STATE {\bf Compute} the solution at $\{r_{k-1}\} \times \mathbb{R}$ using \eqref{solG}.
\ENDFOR
\STATE The desired contingent claim value is the solution at $r_0=0$.
\end{algorithmic}}
\end{algorithm}

\bigskip
In principle, this algorithm can be implemented to mimic the analytical procedure: for instance, in the \textit{Matlab} language, by creating a function handle that stores an association to the Green's function, one can pass it to a quadrature procedure that numerically integrates the function using global adaptive quadrature with default error tolerances. 
However, such a procedure suffers from the curse of dimensionality (and, as a consequence, of the computational times) as the number of jumps increases because of the recursive call to the Green's function.

A possible way to reduce the computational complexity of the semi-analytic procedure is to localize the problem to a bounded domain $[\bar{A},\bar{\bar{A}}]$ and compute formula (\ref{solG}) by means of the trapezoidal rule on grid values of the spot rates. In Appendix \ref{Localization} we detail a strategy for choosing $\{\bar{A},\bar{\bar{A}}\}$, which can be useful also in the finite difference approach.
The theoretical properties of the parabolic operator ensure that the errors introduced by the domain truncation tend to zero as $\bar{A}\to -\infty$, $\bar{\bar{A}}\to  +\infty$ and do not affect the results in an inner region of {observation}\footnote{{The interval $[x_{\min},x_{\max}]$ in the numerical examples below is chosen to highlight the solution behavior. It corresponds to the region where we want to achieve small approximation errors. Nevertheless, in some figures, such as Figure \ref{FigureZCB}, we consider an even larger spatial domain in order to illustrate the behavior of the pricing function more clearly.}}, say $[x_{\min},x_{\max}]$. As an example, we consider the extended Vasicek model, where the PDE between jumps is
\begin{equation}\label{VasPDE}
    {\partial}_t f (t, x) +  \left (\alpha + \beta x \right)  {\partial}_x f (t, x) + \frac{1}{2}  \sigma^2 \, {\partial_{xx}^2}  f   (t, x) -  f (t,x) x = 0.
\end{equation}
Figure \ref{SemiAnalytic} shows the price at $t=0$ of a { ZCB} expiring at $T=1$, with $\alpha = 0.075$, $\beta = -0.3$, $\sigma =0.1$, when one jump $\xi_1 \sim \cN(0.09,0.25)$ in the interest rate occurs at time $t=0.5$. In the first panel, the blue line represents the semi-analytical solution while the red line represents the ZCB price as given by formula (\ref{ZCBpriceGeneral}); the second and third panels show the relative error between the two.
It is evident that the error is relevant (and nevertheless small) only at the truncation points $\bar{A}=-4,\bar{\bar{A}}=4$, due to the truncation of the integral in (\ref{solG}) and in the jump condition. Moreover, it is very small (less than $10^{-11}$) in the inner region of interest rates  $[-0.5, 0.5]$.
Nevertheless, the main drawback of the semi-analytical method is that it requires knowledge of the Green's function, which is not always available.

\begin{center}
	\includegraphics[width=1.05\textwidth]{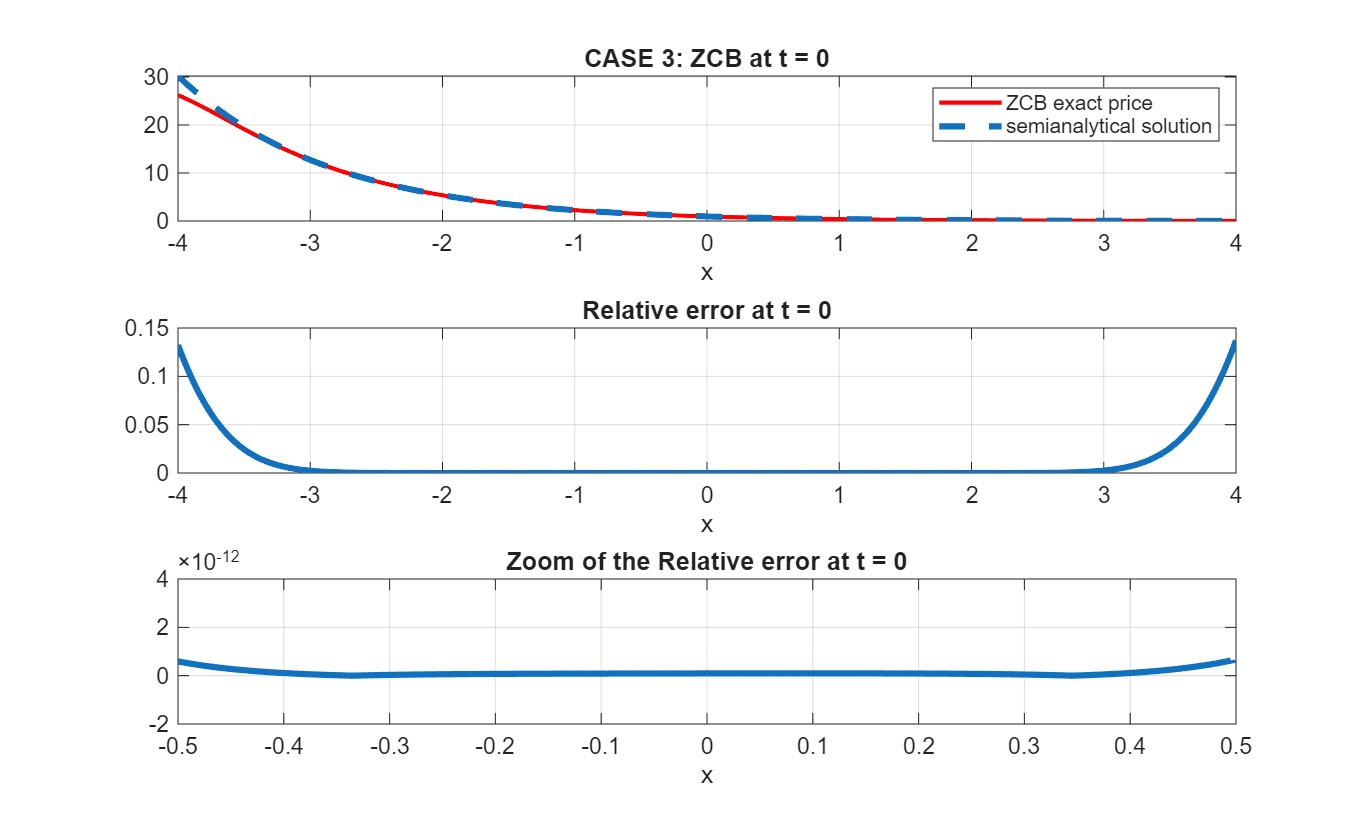}\captionof{figure}{ZCB price at $t=0$ and relative error {(on the whole localized domain $[\bar{A},\bar{\bar{A}}]$ and zoomed on the observation region $[-0.5,0.5]$)} of the semi-analytic method in the Vasicek model, when one jump on the interest rate occurs at time $t=0.5$, with Gaussian distribution $\cN(0.09,0.25)$. Parameter values: $\alpha = 0.075$, $\beta = -0.3$, $\sigma =0.1$, $T=1$.}\label{SemiAnalytic}
\end{center}

\subsection{Finite difference method}\label{FDM}

To illustrate the finite difference approach, without being reductive, we again consider the extended Vasicek model.
In order to numerically solve the backward PDE \eqref{VasPDE} in the unbounded domain $[r_{k-1},r_k) \times \mathbb{R}$, we first localize it to a bounded domain $[r_{k-1},r_k) \times [\bar{A},\bar{\bar{A}}]$. This leads us to define some artificial boundary conditions at $x=\bar A$ and $x=\bar{\bar{A}}$. 
For instance, we may impose on the boundary the final condition value on each time interval, that is, $f(t,\bar A)=f(r_{k}^-,\bar A)$ and $f(t,\bar{\bar{A}})=f(r_{k}^-,\bar{\bar{A}})$, {or any other kind of Neumann boundary condition.}

Another critical point to consider is the computation of the final condition in the presence of a jump on the  { RFR} at time $r_k \in {\cal S}$. Computing the integral term at a point $x$ requires the knowledge of $f(r_k,\cdot)$ on the whole support of ${Q}_{m(k)}$. Here the issue is more delicate, because the integral affects the numerical solution on the entire domain and not only at the artificial boundary, due to the non-local nature of this final condition.
We proceed as follows. 

(i) In the case of an absolutely continuous jump distribution, we decompose the integral as
$$f(r_k^-,x)=\int_{\mathbb{R}} f(r_k,x+z) \, {Q}_{m(k)}(\dd z)=\int_{\mathbb{R}} f(r_k,x+z)\varphi_{m(k)}(z) \, \dd z=$$
$$=\int_{-\infty}^{\bar A} f(r_k,\xi)\varphi_{m(k)}(\xi-x) \, \dd\xi+\int_{\bar A}^{\bar{\bar{A}}} f(r_k,\xi)\varphi_{m(k)}(\xi-x) \, \dd\xi+\int_{\bar{\bar{A}}}^{+\infty} f(r_k,\xi)\varphi_{m(k)}(\xi-x)\, \dd\xi,$$
where $\varphi_{m(k)}$ is the probability density function associated to ${Q}_{m(k)}$. The integral over the computational domain $[\bar A,\bar{\bar{A}}]$ is computed using the trapezoidal rule from the nodal values of the numerical solution at $r_k$, while the two generalized integrals are neglected, taking into account the fast decaying tail behavior of the probability density $\varphi_{m(k)}$ of the jumps.
Obviously, the choice of $\bar A$, $\bar{\bar{A}}$ should also ensure that, for the points $x$ within the region of {observation} $[x_{\text{min}}, x_{\text{max}}]$, the support of the density function $\varphi_{m(k)}(\xi-x)$ is mostly contained within \([\bar{A}, \bar{\bar{A}}]\). A possible choice for $\bar{A}$ and $\bar{\bar{A}}$ is discussed in Appendix \ref{Localization} for the Vasicek model.

(ii) In the case of discrete distribution of spot rates, the discretization step $\Delta x$ in the finite difference scheme should be chosen such that $m_j=c_j \Delta x$ for some integer $c_j$, $j=1,\ldots,M$. This is possible if we assume, for instance, that the jumps on the spot rate are multiples of a fixed amount as a consequence of central banks decisions.
The computation of the final condition at node $x_i$ in the presence of a discrete-distributed jump on the  { RFR} at time $r_k \in {\cal S}$ is performed by setting
$$f(r_k^-,x_i)=f(r_k,x_{i+c_k})p_j+f(r_k,x_{i-c_k})(1-p_j).$$
When either $x_{i+c_k}$ or $x_{i-c_k}$ falls outside the computational domain, we take the value at the boundary.

All this said, {the discretization of \eqref{VasPDE} by $\theta$-method in time and centered finite differences in space over the domain $[r_{k-1},r_k) \times (\bar A,\bar{\bar{A}})$} reads as follows
\begin{equation}\label{FDsyst}
\begin{array}{c}
    \displaystyle\frac{V_i^{j+1}-V_i^j}{\Delta t_k} + 
\frac{1}{2}\sigma^2\theta \frac{V_{i-1}^{j+1}-2V_i^{j+1}+V_{i+1}^{j+1}}{\Delta x^2} + \frac{1}{2}\sigma^2(1-\theta) \frac{V_{i-1}^{j}-2V_i^{j}+V_{i+1}^{j}}{\Delta x^2} +\\
\displaystyle\theta (\alpha+\beta x_i) \frac{V_{i+1}^{j+1}-V_{i-1}^{j+1}}{2\Delta x} + (1-\theta) (\alpha+\beta x_i) \frac{V_{i+1}^{j}-V_{i-1}^{j}}{2\Delta x} -
\theta x_i V_i^{j+1} - (1-\theta) x_i V_i^j=0, 
\end{array}
\end{equation}
for $i=1,2,\ldots,N-1$, $j=M_k-1,M_k-2,\ldots,1,0$ and $\Delta t_k= (r_k-r_{k-1})/M_{k}$, $\Delta x= (\bar{\bar{A}}-\bar A)/N$.
We note that, since this linear algebraic system of equations is solved backward in time, the choice $\theta=1$ yields the \textit{explicit Euler method}, while $\theta=0$ yields the \textit{implicit Euler method} and $\theta=1/2$ yields the \textit{Crank-Nicolson method}.

We come back to the issue of the artificial boundary conditions imposed at the ends $\bar A,\bar{\bar{A}}$ of the short-rate interval. 
Following \cite{TavellaRandall} (p. 123), 
{since our aim is to develop a numerical method that can accommodate different contingent claim valuation problems, we do not impose Neumann or Dirichlet boundary conditions. Instead, we enforce the PDE itself as a boundary condition and employ non-centered finite difference approximations for all the derivatives located at $x_0=\bar A$ and $x_N=\bar{\bar{A}}$. This leads to the addition of the following equations
$$\frac{V_0^{j+1}-V_0^j}{\Delta t_k} + 
\frac{1}{2}\sigma^2\theta \frac{V_{0}^{j+1}-2V_1^{j+1}+V_{2}^{j+1}}{\Delta x^2} + \frac{1}{2}\sigma^2(1-\theta) \frac{V_{0}^{j}-2V_1^{j}+V_{2}^{j}}{\Delta x^2} +$$
$$\theta (\alpha+\beta x_0) \frac{V_{1}^{j+1}-V_{0}^{j+1}}{\Delta x} + (1-\theta) (\alpha+\beta x_0) \frac{V_{1}^{j}-V_{0}^{j}}{\Delta x} -
\theta x_0 V_0^{j+1} - (1-\theta) x_0 V_0^j=0,$$ 
and
$$\frac{V_N^{j+1}-V_N^j}{\Delta t_k} + 
\frac{1}{2}\sigma^2\theta \frac{V_{N-2}^{j+1}-2V_{N-1}^{j+1}+V_{N}^{j+1}}{\Delta x^2} + \frac{1}{2}\sigma^2(1-\theta) \frac{V_{N-2}^{j}-2V_{N-1}^{j}+V_{N}^{j}}{\Delta x^2} +$$
$$\theta (\alpha+\beta x_N) \frac{V_{N}^{j+1}-V_{N-1}^{j+1}}{\Delta x} + (1-\theta) (\alpha+\beta x_N) \frac{V_{N}^{j}-V_{N-1}^{j}}{\Delta x} -
\theta x_N V_N^{j+1} - (1-\theta) x_N V_N^j=0.$$
We observe that the present scheme discretizes both the drift and volatility terms using \mbox{first-order} \mbox{one-sided} finite difference operators. Nevertheless, \mbox{higher-order} discretizations, such as \mbox{second-order} difference schemes, may also be considered\footnote{{In Section~\ref{App:ExUniq}, we proved that the solution may exhibit exponential growth. Obviously, alternative boundary conditions could be imposed, such as, for instance, assuming a linear dependence of the solution on the short-rate. By enlarging the computational domain, boundary approximation errors do not affect the numerical solution in the inner region of interest.}}.}

\subsection{Zero-coupon bond}\label{ZCBsimul}

In this section, we validate our PDE formulation for pricing { ZCBs}   and analyze the impact of stochastic discontinuities on their valuation. We focus on the Vasicek model, as it offers closed-form solutions and an analytical expression for the Green's function, making it well-suited for theoretical and numerical comparison.
Figure \ref{FigureZCB} shows the effect of jumps on the ZCB price in the following scenarios:
\begin{itemize}
    \item CASE 1: classical Vasicek model (Parameter values: $\alpha = 0.075$, $\beta = -0.3$, $\sigma =0.1$, $T =1$);
    \item CASE 2: Vasicek model + one jump on the num\'eraire occurring at time $t=0.8$;
    \item CASE 3: Vasicek model + one Gaussian jump $\cN(m,\gamma^2)$ on the interest rate at time $t=0.5$, with $m=0.09$ and $\gamma=0.5$;
    \item CASE 4: Vasicek model + one jump on the num\'eraire at time $t=0.8$ + one Gaussian jump $\cN(m,\gamma^2)$ on the interest rate at time $t=0.5$, with $m=0.09$ and $\gamma=0.5$;
\end{itemize}
\begin{center}\vspace{-0.35cm}
	\includegraphics[width=0.78\textwidth]{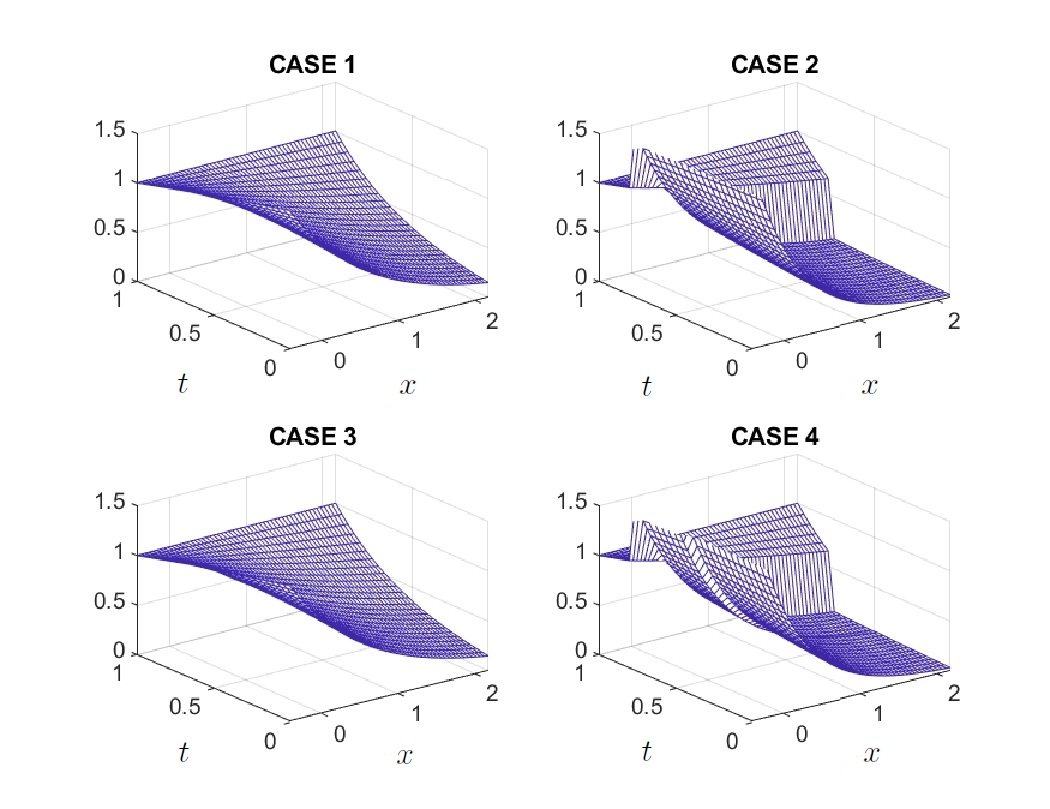}\vspace{-0.5cm}\captionof{figure}{ZCB price. Parameter values: $\alpha = 0.075$, $\beta = -0.3$, $\sigma =0.1$, $T =1$, $\mu=a/b$.}\label{FigureZCB}
\end{center}
We observe that, while the effect of the jumps on the num\'eraire is evident, the jump in the RFR results in only a minor discontinuity in the ZCB price in CASE 3, becoming more pronounced in CASE 4.

Suppose that we are interested in pricing ZCBs when the spot rate ranges within the interval $[-0.5,1]$. The finite difference approximation to the closed-form formula (\ref{ZCBpriceGeneral}) is obtained by the Crank-Nicolson scheme with $\Delta x= 0.005$ and {$\Delta t= 0.002$}. The computational domain selected by the procedure described in Appendix \ref{Localization} is $[-1.6,2.065]$ in CASE 1 and CASE 2 and a bit larger, namely $[-5.1204,5.6196]$ in CASE 3 and CASE 4.

{To illustrate the effect of the imposed boundary conditions and their impact on the quality of the numerical solution in the inner region $[-0.5,1]$, we plot in Figure \ref{Errori} the absolute error between the exact and numerical solutions over the entire localized spatial domain for Cases 2 and 3. As previously stated, the errors at the boundary are rapidly damped and do not affect the quality of the approximate solution in the inner region.}
\begin{center}
	\includegraphics[scale=0.35]{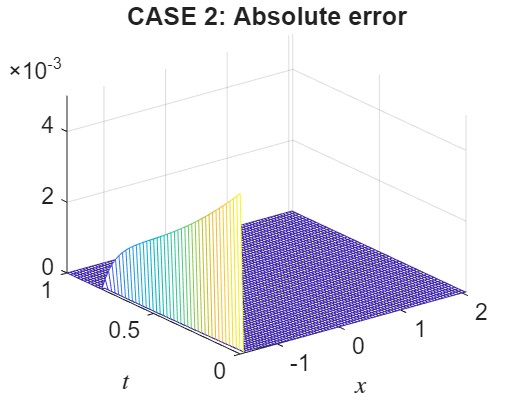}
    \includegraphics[scale=0.35]{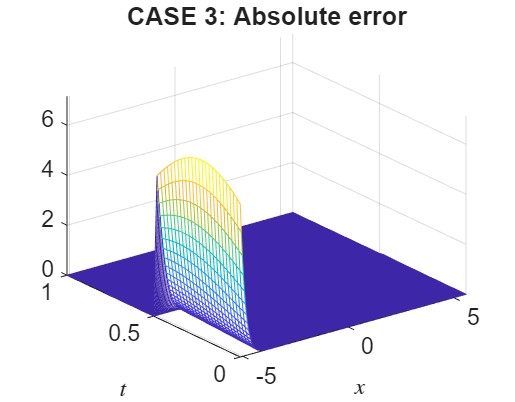}
    \captionof{figure}{{Absolute error between the exact and numerical solutions over the entire localized spatial domain for Cases 2 and 3. Parameter values: $\alpha = 0.075$, $\beta = -0.3$, $\sigma =0.1$, $T =1$.}}\label{Errori}
\end{center}
In all cases, the correctness of the PDE formulation and the precision of the discretization scheme ({up to an overall order of $10^{-6}$ or $10^{-7}$ in the observation region $[-0.5,1]$}) are evident in Figures \ref{ZCB_CASO2}, \ref{ZCB_CASO3}, and \ref{ZCB_CASO4}, where the exact and numerical solutions are plotted together with their absolute and relative errors. {We observe that the exact and numerical solutions are visually indistinguishable to the human eye.}
\begin{center}\vspace{-0.1cm}
    \includegraphics[width=0.9\textwidth]{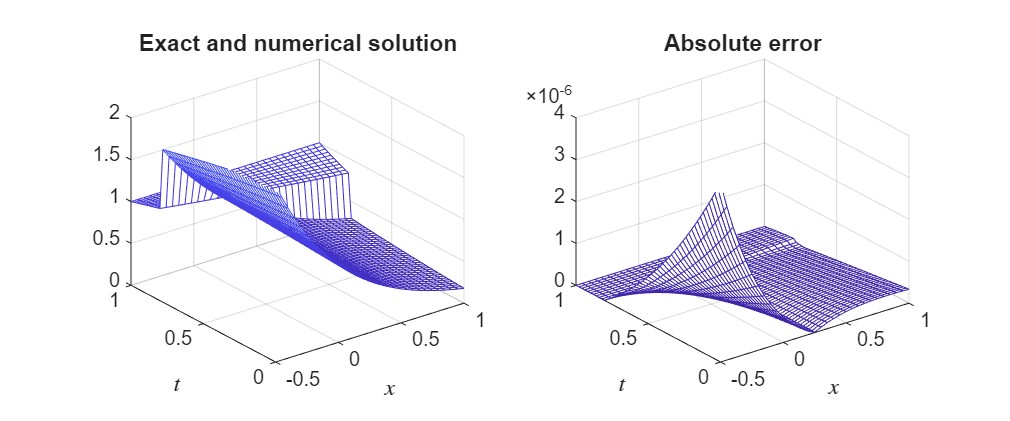}
    \vspace{-0.2cm}
    \captionof{figure}{ZCB Gaussian CASE 2. Parameter values: $\alpha = 0.075$, $\beta = -0.3$, $\sigma =0.1$, $T =1$.}\vspace{-0.1cm}\label{ZCB_CASO2}
\end{center}

\begin{center}
    \includegraphics[width=0.9\textwidth]{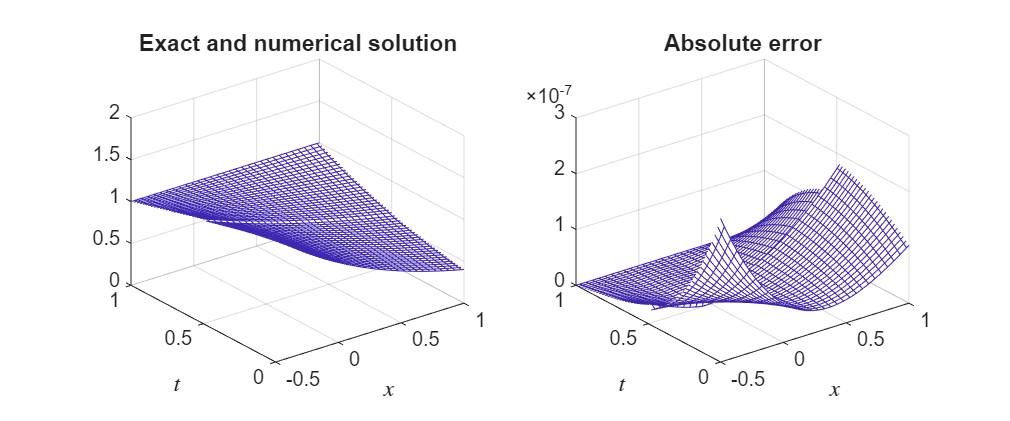}
    \vspace{-0.2cm}
    \captionof{figure}{ZCB Gaussian CASE 3. Parameter values: $\alpha = 0.075$, $\beta = -0.3$, $\sigma =0.1$, $T =1$.}\vspace{-0.1cm}\label{ZCB_CASO3}
\end{center}
\begin{minipage}[h]{17.cm}
    \begin{center}
    \includegraphics[width=0.9\textwidth]{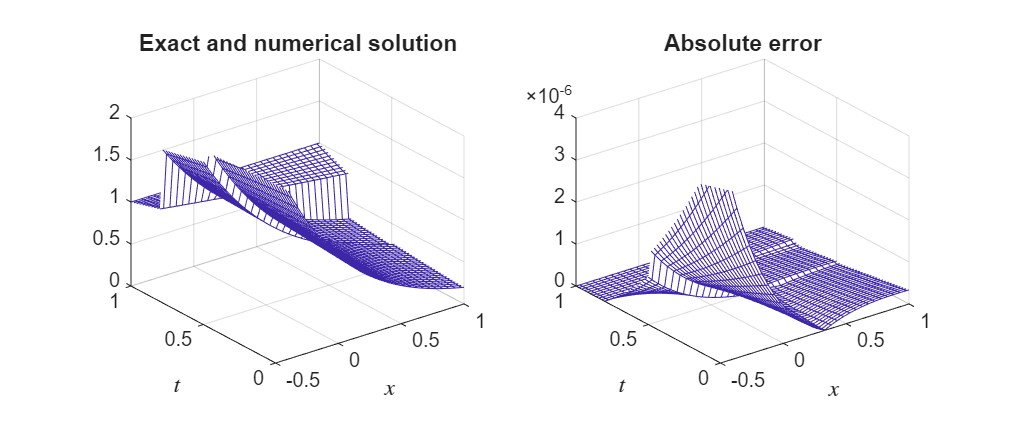}
    \vspace{-0.2cm}
    \captionof{figure}{ZCB Gaussian CASE 4. Parameter values: $\alpha = 0.075$, $\beta = -0.3$, $\sigma =0.1$, $T =1$.}\label{ZCB_CASO4}
\end{center}
\end{minipage}\\

Table \ref{tab:ZCBgaussian} shows the maximum (in time) average absolute error
$$
\textrm{Abs. Err.}:=\max_{\scriptsize\begin{array}{l}
    j=0,\ldots M_k-1\\
    k=0.\ldots, K
\end{array}}\mathop{\textrm{mean}}_{ x_i\in[-0.5,1]}|V_i^j-P_T(t_j,x_i)|
$$
and relative error w.r.t. the exact solution in \eqref{ZCBpriceGeneral} of the finite difference approximation and of the semi-analytic solution for the ZCB in the inner region $[-0.5,1]$ of {observation} in CASE 2 and CASE 4 with Gaussian jump distribution, for different choices of the discretization step $\Delta x$ (and {$\Delta t= 0.002$}). 
The error reduction in the finite difference approximation shows {second-order convergence (in space) of the numerical method} to the closed-form solution, while it is evident that the errors in the semi-analytic method are very small and remain relatively stable, with only slight variations attributable to floating-point arithmetic computations.
Results for CASE 3 are analogous to CASE 4 and thus not reported.
\begin{minipage}[h]{17.5cm}
\vspace{0.5cm}
\begin{center}
\begin{small}
    \begin{tabular}{l|c|c|c|c|c|c|c|c}
    \hline
      &  \multicolumn{4}{c|}{\textbf{CASE 2}} & \multicolumn{4}{c}{\textbf{CASE 4}}\\
    \hline \hline
     & \multicolumn{2}{c|}{\textbf{Finite Diff. Method}} & \multicolumn{2}{c|}{\textbf{Semi-analytic method}} & \multicolumn{2}{c|}{\textbf{Finite Diff. Method}} & \multicolumn{2}{c}{\textbf{Semi-analytic method}}\\
    \hline
     $\Delta x$ & \textbf{Abs. Err.}  & \textbf{Rel. Err.} & \textbf{Abs. Err.} & \textbf{Rel. Err.} & \textbf{Abs. Err.}  & \textbf{Rel. Err.} & \textbf{Abs. Err.} & \textbf{Rel. Err.}\\
     \hline
1e-2 & 3.45e-6 & 4.02e-6 & 9.70e-15 & 1.16e-14 & 4.25e-6 & 4.08e-6 & 1.19e-13 & 1.56e-13\\
5e-3 & 8.59e-7 & 9.78e-7 & 7.34e-15 & 1.03e-14 & 1.07e-6 & 9.93e-7 & 9.62e-14 & 1.22e-13\\
2.5e-3 & 2.06e-7 & 2.15e-7 & 7.80e-15 & 1.06e-14 &  2.63e-7 & 2.22e-7 & 1.17e-13 & 1.47e-13\\
1.25e-3 & 4.57e-8 & 3.71e-8 & 8.24e-15 & 1.12e-14 & 6.98e-8 & 5.68e-8 & 1.29e-13 &  1.62e-13
    \end{tabular}
    \captionof{table}{ZCB in CASE 2 and CASE 4 with Gaussian jump distribution. Average absolute and relative error of the finite difference approximation and of the semi-analytic solution for the {ZCB} in the inner region $[-0.5,1]$.}
    \label{tab:ZCBgaussian}
\end{small}
\end{center}\vspace{0.25cm}
\end{minipage}\\
Similar results are obtained in the presence of a discrete jump distribution for the spot interest rate. We replicate scenarios similar to Cases 3 and 4, featuring a discretely distributed jump in the RFR, as outlined in Section~\ref{sec:simulation}.
Figure \ref{ZCB_Bernoulli4} shows the finite difference approximation to the closed-form formula (\ref{ZCBpriceGeneral}) in the computational domain $[-1.60,2.065]$ with $\Delta x= 0.005$ and {$\Delta t= 0.002$} for CASE 4. 
\begin{center}
\includegraphics[width=0.9\textwidth]{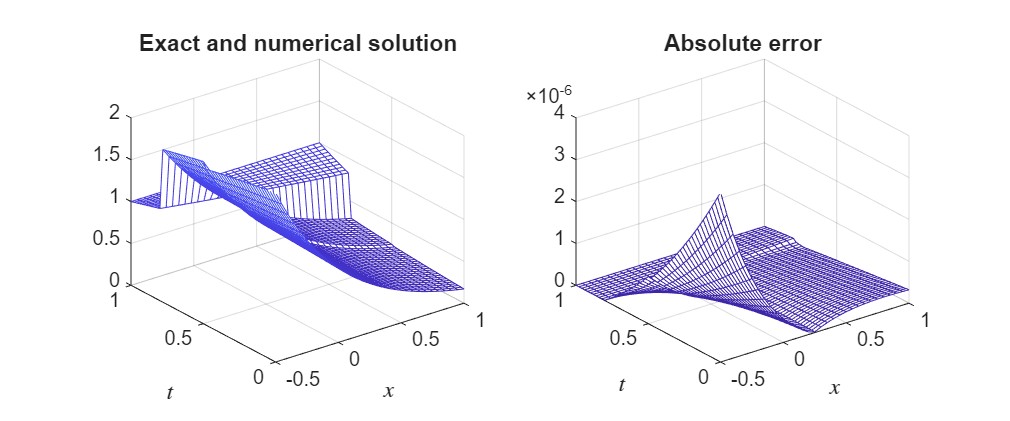}
\captionof{figure}{ZCB CASE 4 with discrete jumps. Parameter values: $\alpha = 0.075$, $\beta = -0.3$, $\sigma =0.1$, $T =1$, $m=0.09$, $p=0.7$.}\label{ZCB_Bernoulli4}
\end{center}

Table \ref{tab:ZCBdiscrete} shows the maximum (in time) average absolute and relative error of the finite difference approximation and of the semi-analytic solution for the ZCB in the inner region $[-0.5,1]$ in CASE 4 with discrete jump distribution. The same comments on the convergence of the numerical method and on the error stability of the semi-analytic method hold.
These numerical results testify to the validity of our PDE-based model for pricing a {ZCB}.
\begin{table}[]
    \centering
    \begin{tabular}{c|c|c|c|c}
    \hline
      &  \multicolumn{4}{c}{\textbf{CASE 4}}\\
    \hline \hline
     & \multicolumn{2}{c|}{\textbf{Finite Diff. Method}} & \multicolumn{2}{c}{\textbf{Semi-analytic method}}\\
    \hline
     $\Delta x$ & \textbf{Abs. Err.}  & \textbf{Rel. Err.} & \textbf{Abs. Err.} & \textbf{Rel. Err.}\\
     \hline
1e-2 & 3.27e-06 & 4.02e-06 & 9.84e-15 & 1.15e-14\\
5e-3 & 8.14e-07 & 9.75e-07 & 7.46e-15 & 1.02e-14\\
2.5e-3 & 1.95e-07 &  2.12e-07 & 7.94e-15 & 1.05e-14\\
1.25e-3 & 4.35e-08 & 3.83e-08 & 8.38e-15 &   1.11e-13
    \end{tabular}
    \caption{ZCB in CASE 4 with discrete jump distribution. Average absolute and relative error of the finite difference approximation and of the semi-analytic solution for the {ZCB} in the inner region $[-0.5,1]$.}
    \label{tab:ZCBdiscrete}
\end{table}

{We conclude this section with a test assessing the performance of our numerical procedure as the number of jumps $K$ increases, which has important implications for use in future calibration. 
We consider separately the cases in which only the number of jumps in the num\'eraire increases and those in which only the number of jumps in the interest rate increases, in order to highlight the different behaviors in terms of computational speed and accuracy.
The results are reported in Table \ref{tab:PiuSalti}.}
{We observe that, in CASE 3 with discrete jumps, both the discretization errors and the CPU time remain very stable as $K$ increases. In CASE 3 with Gaussian jumps, the errors slightly deteriorate due to the approximate evaluation of the jump condition at each time level $t_j$.}
 {Different considerations apply to CASE 2. In this case, the deterioration in the quality of the approximation is due to the jump condition, which shifts the ZCB price to a higher level at each jump time $t_j$.  Second-order convergence is still maintained; however, achieving higher accuracy in the region of observation $[-0.5,1]$ requires a reduction in $\Delta x$. This is due to the scaling condition associated with jumps in the num\'eraire: the scaling factor $\de^{-x}$ can be large on the left-hand side of the computational domain $[-1.6, 2.065]$, thereby amplifying any approximation errors affecting the numerical solution at $t_{j}^+$. 
} 
{This effect would be reduced if we were interested only in computing ZCB prices for positive spot rates. Indeed, in this case the absolute error for four jumps remains as low as $9.9\cdot 10^{-7}$.}
\begin{table}[]
    \centering
    \begin{tabular}{c|c|c|c|c|c|c|c|c|c}
    \hline
      &  \multicolumn{3}{|c|}{\textbf{CASE 3 Gaussian}} & \multicolumn{3}{|c|}{\textbf{CASE 3 Discrete}} & \multicolumn{3}{|c}{\textbf{CASE 2}}\\
          \hline \hline
     $K$ & \textbf{Abs. Err.}  & \textbf{$\%$ Err.} & \textbf{CPU} & \textbf{Abs. Err.} & \textbf{$\%$ Err.} & \textbf{CPU} & \textbf{Abs. Err.} & \textbf{$\%$ Err.} & \textbf{CPU}\\
     \hline
1 & 7.49e-8 & 8.89e-8 & 2.16 &   5.83e-8 & 6.28e-8 & 0.09 & 8.59e-7 &   9.78e-7 & 0.11  \\
2 & 1.02e-7 & 1.33e-7 & 2.86 &   5.99e-8 & 6.75e-8 & 0.19 & 3.59e-6 &   3.71e-6 & 0.12  \\
3 & 1.30e-7 & 1.78e-7 & 2.95 &   6.17e-8 & 7.36e-8 & 0.22 & 9.47e-6 &   7.95e-6 & 0.17  \\
4 & 1.86e-7 & 2.48e-7 & 3.83 &   6.26e-8 & 7.96e-8 & 0.28 & 1.88e-5 &   1.19e-5 & 0.17  \\
    \end{tabular}
    \caption{ZCB for increasing number of jumps $K$. Absolute and relative error at $t=0$ and CPU time (in seconds) for the finite difference approximation in the inner region $[-0.5,1]$. Discretization steps: $\Delta x= 0.005$ and $\Delta t= 0.002$.}
    \label{tab:PiuSalti}
\end{table}

\subsection{Call option on $P_T(t, x)$}

The derivation of the closed-form price of a call option on a ZCB under the assumption that the short-rate follows Vasicek-type dynamics with  Gaussian stochastic discontinuities is detailed in Proposition \ref{affinecallprop}. From this point on, we focus our attention on the accuracy of the numerical approximation method, as the semi-analytic method yields results with an accuracy of the order of $1\de-10$ or even smaller. Therefore, the semi-analytic method perfectly replicates the closed-form solution, when available, and can serve as a valid and easier to implement substitute for the closed-form price when the Green's function associated with the PDE in (\ref{eq:lastPDE}) is available.

We simulate the same scenarios as in Section \ref{ZCBsimul}, with the same parameter values. The strike price is $K=0.5$, while the maturity of the ZCB is at $S=1.5$.
Again, the finite difference approximation to the closed-form formula (\ref{affinecallpricefinal}) is obtained by the Crank-Nicolson scheme in the same computational domain as in Section \ref{ZCBsimul}, with $\Delta x= 0.005$ and {$\Delta t= 0.002$}. 
The correctness of the PDE formulation and the precision of the discretization scheme are evident in Figures 
\ref{CallCASO2Gauss}, \ref{CallCASO3Gauss} and \ref{CallCASO4Gauss}, where the exact and numerical solutions are plotted together with their absolute and relative errors.

\begin{center}
    \includegraphics[scale=0.4]{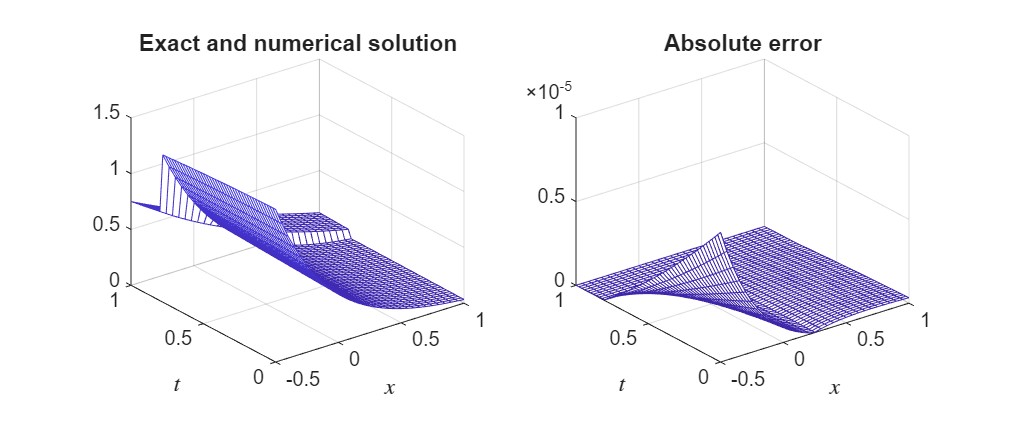}
    \captionof{figure}{Call option, CASE 2. Parameter values: $\alpha = 0.075$, $\beta = -0.3$, $\sigma =0.1$, $K=0.5$, $T =1$, $S=1.5$.}\label{CallCASO2Gauss}
\end{center}

\begin{center}
	\includegraphics[scale=0.4]{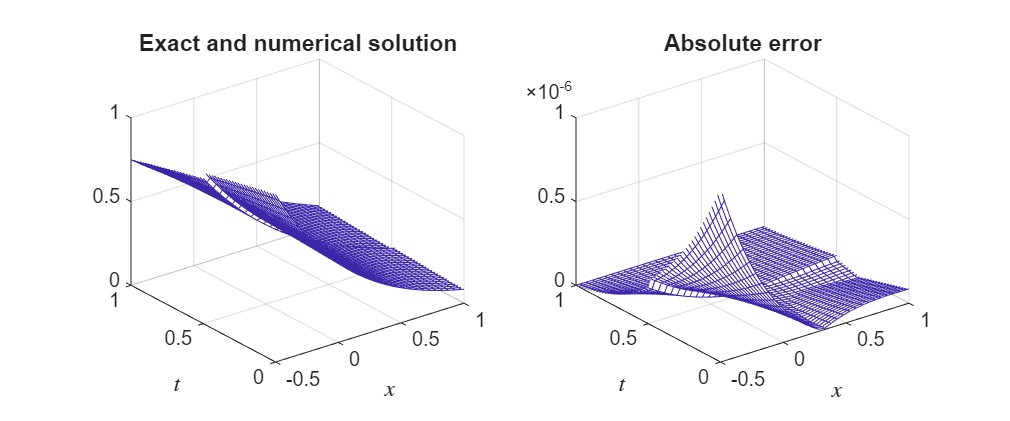}
    \captionof{figure}{Call option, CASE 3 with Gaussian jump. Parameter values: $\alpha = 0.075$, $\beta = -0.3$, $\sigma =0.1$, $m=0.09$, $\gamma=0.5$, $K=0.5$, $T =1$, $S=1.5$.}\label{CallCASO3Gauss}
\end{center}

\begin{center}
	\includegraphics[scale=0.4]{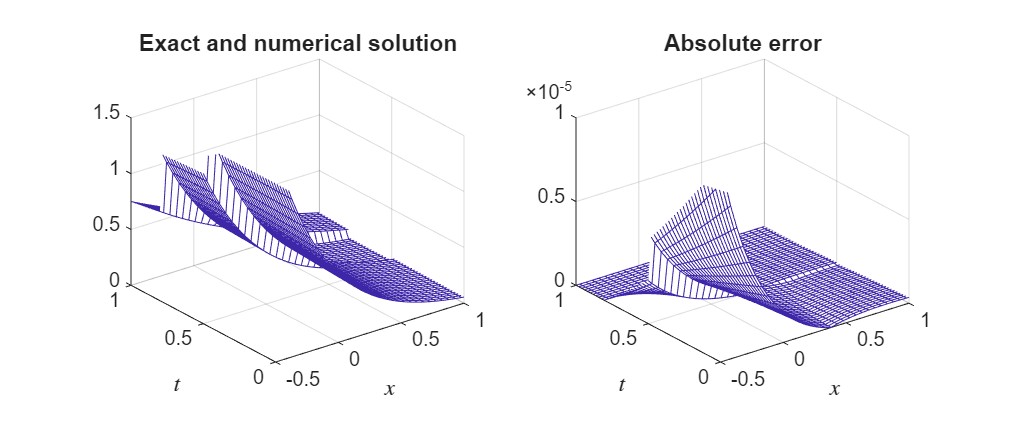}
    \captionof{figure}{Call option, CASE 4 with Gaussian jump. Parameter values: $\alpha = 0.075$, $\beta = -0.3$, $\sigma =0.1$, $m=0.09$, $\gamma=0.5$, $K=0.5$, $T =1$, $S=1.5$.}\label{CallCASO4Gauss}
\end{center}

The corresponding maximum (in time) average absolute and relative errors of the finite difference approximation in the inner region $[-0.5,1]$ for decreasing $\Delta x$ and {$\Delta t= 0.002$} are shown in Table \ref{tab:CallOption}. In all the examined cases, the numerical method converges to the exact price of the option.
\begin{table}[ht]
    \small\centering
    \begin{tabular}{c|c|c|c|c|c|c|c|c}
    \hline
      &  \multicolumn{2}{c|}{\textbf{CASE 1}} & \multicolumn{2}{c|}{\textbf{CASE 2}} & \multicolumn{2}{c|}{\textbf{CASE 3}} & \multicolumn{2}{c}{\textbf{CASE 4}}\\
    \hline \hline
     $\Delta x$ & \textbf{Abs. Err.}  & \textbf{Rel. Err.} & \textbf{Abs. Err.} & \textbf{Rel. Err.} & \textbf{Abs. Err.}  & \textbf{Rel. Err.} & \textbf{Abs. Err.} & \textbf{Rel. Err.}\\
     \hline
1e-2 & 9.11e-7 & 2.68e-6 & 5.40e-6 & 1.45e-5 & 9.46e-7 & 2.41e-6 & 8.54e-6 & 1.27e-5 \\
5e-3 & 2.26e-7 & 6.56e-7 & 1.36e-6 & 3.62e-6 & 2.29e-7 & 5.17e-7 & 2.15e-6 & 3.16e-6 \\
2.5e-3 & 5.37e-8 &  1.49e-7 & 3.38e-7 & 9.03e-7 & 5.88e-8 & 1.33e-7 & 5.42e-7 & 7.90e-7 \\
1.25e-3 & 1.10e-8 & 2.29e-8 & 8.28e-8 &   2.19e-7 & 1.54e-8& 3.04e-8 & 1.37e-7 & 1.97e-7 
    \end{tabular}
    \caption{Call option on ZCB with Gaussian jumps. Average absolute and relative error of the finite difference approximation in the inner region $[-0.5,1]$.}
    \label{tab:CallOption}
\end{table}

Finally, the case of a call option on ZCB with discrete jumps is considered.
The parameter values are the same as those mentioned above.
In this case, the price of the derivative is no longer available in \mbox{closed-form.} Therefore, we assume the semi-analytic solution as a benchmark to validate the finite difference one.
The convergence of the finite difference approximation is evident in Table \ref{tab:CallOptionDiscrete}, where the average absolute and relative errors computed in the inner region $[-0.5,1]$ with respect to the semi-analytic solution at $t = 0$ are shown. 
\begin{table}[ht]
    \centering
    \begin{tabular}{c|c|c|c|c}
    \hline
      &  \multicolumn{2}{c|}{\textbf{CASE 3}} & \multicolumn{2}{c}{\textbf{CASE 4}}\\
    \hline \hline
     $\Delta x$ & \textbf{Abs. Err.}  & \textbf{Rel. Err.} & \textbf{Abs. Err.} & \textbf{Rel. Err.}\\
     \hline
1e-2 & 8.77e-07 & 2.74e-06 & 5.04e-06 & 1.47e-05 \\
5e-3 & 2.18e-07 & 6.72e-07 & 1.27e-06 & 3.68e-06 \\
2.5e-3 & 5.19e-08 & 1.54e-07 & 3.16e-07  & 9.18e-07 \\
1.25e-3 & 1.08e-08 & 2.54e-08 & 7.75e-08 & 2.23e-07 
    \end{tabular}
    \caption{Call option on ZCB with discrete jumps in Cases 3 and 4. Average absolute and relative error of the finite difference approximation in the inner region $[-0.5,1]$.}
    \label{tab:CallOptionDiscrete}
\end{table}

\section{Conclusions}
This paper develops a PDE approach for pricing interest rate derivatives in the presence of stochastic discontinuities. Within an arbitrage-free framework, we consider a general short-rate model for  { RFR}s that incorporates discontinuities at fixed dates with random sizes, as well as a num\'eraire exhibiting discontinuities at roll-over dates.
First, we show that under suitable assumptions, the price of a \mbox{European-style} derivative can be represented as a function of time and the current RFR, which solves an associated PDE. We also establish a Feynman-Kač representation for the solution.
Next, we specialize the framework to the affine term structure case and derive a closed-form solution for the price of a { ZCB}. For more general derivatives, such as options on { ZCB}s, closed-form solutions are available only under additional assumptions, e.g., in the Hull-White model with Gaussian jumps.
As an alternative, we propose two numerical methods for solving the PDE: (i) a semi-analytic method based on the explicit construction of the Green's function, and (ii) a finite difference scheme. We validate both approaches by comparing the numerical results to closed-form solutions in the Vasicek model with stochastic discontinuities, when available, and then we provide an example of their effectiveness in a case with unknown solution.
Our findings highlight the flexibility and effectiveness of the PDE framework, especially in situations where closed-form formulas are either unavailable or cumbersome to derive. This approach proves particularly promising for pricing more complex contingent claims and for applications in more general short-rate models.
Future research will focus on calibrating the model to market data and extending the methodology to richer term structure models and multi-factor settings. 

\section*{Acknowledgements}{We acknowledge financial support from the Italian Group for Scientific Calculus (INdAM - GNCS Project CUP\textunderscore E53C24001950001) and from the Bank of Italy. We thank Davide Addona, Luca Donati, Claudio Fontana, Gianluca Fusai, Luca Lorenzi, Daniele Marazzina, and Thorsten Schmidt for useful discussions and comments on the paper. We also thank the two anonymous referees and the Associate Editor for their careful review and constructive comments, which helped improve the quality of the manuscript.}

 \bibliographystyle{plainnat}
 \bibliography{Bibliography}

\begin{thebibliography}{43}
\providecommand{\natexlab}[1]{#1}
\providecommand{\url}[1]{\texttt{#1}}
\expandafter\ifx\csname urlstyle\endcsname\relax
  \providecommand{\doi}[1]{doi: #1}\else
  \providecommand{\doi}{doi: \begingroup \urlstyle{rm}\Url}\fi

\bibitem[Anbil et~al.(2020)Anbil, Anderson, and Senyuz]{anbiletal2020}
S.~Anbil, A.~Anderson, and Z.~Senyuz.
\newblock What happened in money markets in september 2019?, 2020.
\newblock URL
  \url{https://www.federalreserve.gov/econres/notes/feds-notes/what-happened-in-money-markets-in-september-2019-20200227.htm}.

\bibitem[Andersen and Bang(2024)]{Andersen02082024}
L.~Andersen and D.~Bang.
\newblock Spike and hike modeling for interest rate derivatives with an
  application to {SOFR} caplets.
\newblock \emph{Quantitative Finance}, 24\penalty0 (8):\penalty0 1017--1033,
  2024.

\bibitem[Angelini et~al.(2025)Angelini, Herzel, and
  Nicolosi]{Angelini2025Modeling}
F.~Angelini, S.~Herzel, and M.~Nicolosi.
\newblock Modeling {E}uro area benchmark rates after the end of {LIBOR}.
\newblock Working paper, 2025.
\newblock URL \url{https://ssrn.com/abstract=5575270}.

\bibitem[Backwell and Hayes(2022)]{BackwellHayes2022}
A.~Backwell and J.~Hayes.
\newblock Expected and unexpected jumps in the overnight rate: Consistent
  management of the {L}ibor transition.
\newblock \emph{Journal of Banking \& Finance}, 145:\penalty0 106669, 2022.

\bibitem[Baldi(2017)]{Baldi}
P.~Baldi.
\newblock \emph{Stochastic Calculus. An Introduction Through Theory and
  Exercises}.
\newblock Universitext. Springer, 2017.

\bibitem[Bickersteth et~al.(2025)Bickersteth, Ding, and
  Rutkowski]{rutkowski2021pricing}
M.~Bickersteth, Y.~Ding, and M.~Rutkowski.
\newblock Pricing and hedging of {SOFR} derivatives.
\newblock \emph{Mathematical Finance}, 2025.
\newblock URL \url{https://onlinelibrary.wiley.com/doi/pdf/10.1111/mafi.70004}.

\bibitem[Bj\"ork(1998)]{bjork:arbitragetheory}
T.~Bj\"ork.
\newblock \emph{Arbitrage Theory in Continuous Time}.
\newblock Oxford Univ. Press, 1998.

\bibitem[Brace et~al.(2024)Brace, Gellert, and Schlögl]{Brace2024}
A.~Brace, K.~Gellert, and E.~Schlögl.
\newblock {SOFR} term structure dynamics—discontinuous short rates and
  stochastic volatility forward rates.
\newblock \emph{Journal of Futures Markets}, 44\penalty0 (6):\penalty0
  936--985, 2024.

\bibitem[B\"uttler and J.(1996)]{Buttler1996}
H.J. B\"uttler and Waldvogel J.
\newblock Pricing callable bonds by means of {G}reen’s function.
\newblock \emph{Mathematical Finance}, 6\penalty0 (1):\penalty0 53--88, 1996.

\bibitem[Chan et~al.(1992)Chan, Karolyi, Longstaff, and
  Sanders]{Chanetal92:CEV}
K.C. Chan, A.J. Karolyi, F.A. Longstaff, and A.B. Sanders.
\newblock An empirical comparison of alternative models of the short-term
  interest rate.
\newblock \emph{Journal of Finance}, 47\penalty0 (3):\penalty0 1209--27, 1992.

\bibitem[Cohen and Elliott(2015)]{cohen:stochcalculus}
S.~N. Cohen and R.~J. Elliott.
\newblock \emph{Stochastic calculus and applications}.
\newblock Probability and its Applications. Springer, Cham, second edition,
  2015.

\bibitem[da~Silva and Baczynski(2024)]{economies12030073}
A.~J. da~Silva and J.~Baczynski.
\newblock Discretely distributed scheduled jumps and interest rate derivatives:
  Pricing in the context of central bank actions.
\newblock \emph{Economies}, 12\penalty0 (3), 2024.

\bibitem[De~Genaro and Avellaneda(2018)]{DeGenaro2018}
A.~De~Genaro and M.~Avellaneda.
\newblock Pricing interest rate derivatives under monetary changes.
\newblock \emph{International Journal of Theoretical and Applied Finance},
  21\penalty0 (6):\penalty0 1850037, 2018.

\bibitem[Filipovic(2009)]{filipovic2009term}
D.~Filipovic.
\newblock \emph{Term-structure models: A graduate course}.
\newblock Springer Science \& Business Media, 2009.

\bibitem[Fontana(2023)]{Fontana2023}
C.~Fontana.
\newblock Caplet pricing in affine models for alternative risk-free rates.
\newblock \emph{SIAM Journal on Financial Mathematics}, 14\penalty0
  (1):\penalty0 SC1--SC16, 2023.

\bibitem[Fontana et~al.(2020)Fontana, Grbac, Gümbel, and
  Schmidt]{fontana2020term}
C.~Fontana, Z.~Grbac, S.~Gümbel, and T.~Schmidt.
\newblock Term structure modelling for multiple curves with stochastic
  discontinuities.
\newblock \emph{Finance and Stochastics}, 24\penalty0 (2):\penalty0 465--511,
  2020.

\bibitem[Fontana et~al.(2024)Fontana, Grbac, and
  Schmidt]{fontanaetal2024:termstruct}
C.~Fontana, Z.~Grbac, and T.~Schmidt.
\newblock Term structure modeling with overnight rates beyond stochastic
  continuity.
\newblock \emph{Mathematical Finance}, 34\penalty0 (1):\penalty0 151--189,
  2024.

\bibitem[Fontana et~al.(2025)Fontana, Pavarana, and Schmidt]{fontanaPS}
C.~Fontana, S.~Pavarana, and T.~Schmidt.
\newblock An extended {CIR} process with stochastic discontinuities.
\newblock 2025.
\newblock URL \url{https://arxiv.org/abs/2509.15752}.

\bibitem[Freidlin(1985)]{Freidlin}
M.~Freidlin.
\newblock \emph{Functional Integration and Partial Differential Equations},
  volume 109 of \emph{Annals of Mathematics Studies}.
\newblock Princeton University Press, Princeton, NJ, USA, 1985.

\bibitem[Friedman(1964)]{Friedman}
A.~Friedman.
\newblock \emph{Partial Differential Equations of Parabolic Type}.
\newblock Prentice-Hall, Englewood Cliffs, NJ, USA, 1964.

\bibitem[Fusai and Gambaro(2024)]{fusai2024}
G.~Fusai and A.~Gambaro.
\newblock Pricing on trees using new risk-free rates.
\newblock \emph{The Journal of Derivatives}, 32\penalty0 (1):\penalty0
  139--159, 2024.

\bibitem[Gellert and Schl\"ogl(2021)]{gellert2021short}
K.~Gellert and E.~Schl\"ogl.
\newblock Short rate dynamics: A {F}ed {F}unds and {SOFR} perspective, 2021.
\newblock URL \url{https://arxiv.org/abs/2101.04308}.

\bibitem[Heitfield and Park(2019)]{heitfield2019inferring}
E.~Heitfield and Y.~Park.
\newblock Inferring term rates from {SOFR} futures prices.
\newblock Finance and Economics Discussion Series 2019-014, Board of Governors
  of the Federal Reserve System (U.S.), Washington, D.C., 2019.
\newblock URL \url{https://doi.org/10.17016/FEDS.2019.014}.

\bibitem[Jacod and Shiryaev(2003)]{jacod2013:limit}
J.~Jacod and A.~N. Shiryaev.
\newblock \emph{Limit theorems for stochastic processes}, volume 288 of
  \emph{Grundlehren der Mathematischen Wissenschaften}.
\newblock Springer-Verlag, Berlin, second edition, 2003.

\bibitem[Karatzas and Shreve(1991)]{karatzas:brown}
I.~Karatzas and S.E. Shreve.
\newblock \emph{Brownian Motion and Stochastic Calculus}.
\newblock Springer, New York, 2nd edition, 1991.

\bibitem[Keller-Ressel et~al.(2019)Keller-Ressel, Schmidt, and
  Wardenga]{KellerRessel2019}
M.~Keller-Ressel, T.~Schmidt, and R.~Wardenga.
\newblock Affine processes beyond stochastic continuity.
\newblock \emph{Annals of Applied Probability}, 29\penalty0 (6):\penalty0
  3387--3437, 2019.

\bibitem[Kim and Wright(2014)]{kim2014jumps}
D.~H. Kim and J.~H. Wright.
\newblock Jumps in bond yields at known times.
\newblock {NBER} working paper, National Bureau of Economic Research, 2014.
\newblock URL \url{https://www.nber.org/papers/w20409}.

\bibitem[Klingler and Syrstad(2021)]{klingler2021life}
S.~Klingler and O.~Syrstad.
\newblock Life after {LIBOR}.
\newblock \emph{Journal of Financial Economics}, 141\penalty0 (2):\penalty0
  783--801, 2021.

\bibitem[Liu and Song(2024)]{liu2024risk}
F.~Liu and Y.~Song.
\newblock Risk-free rate caplets pricing by {CTMC} approximation.
\newblock \emph{Quantitative Finance}, 24\penalty0 (11):\penalty0 1579--1595,
  2024.

\bibitem[Lo(2013)]{Lo2013}
C.F. Lo.
\newblock Lie-algebraic approach for pricing zero-coupon bonds in single-factor
  interest rate models.
\newblock \emph{Journal of Applied Mathematics}, \penalty0 (1):\penalty0
  276238, 2013.

\bibitem[Lunardi(1995)]{Lunardi}
A.~Lunardi.
\newblock \emph{Analytic Semigroups and Optimal Regularity in Parabolic
  Problems}, volume~16 of \emph{Progress in Nonlinear Differential Equations
  and Their Applications}.
\newblock Birkhauser, Basel, Switzerland, 1995.

\bibitem[Lunardi(1998)]{Lunardi1998}
A.~Lunardi.
\newblock Schauder theorems for linear elliptic and parabolic problems with
  unbounded coefficients in $\mathbb{R}^n$.
\newblock \emph{Studia Mathematica}, 128\penalty0 (2):\penalty0 171--198, 1998.

\bibitem[Lyashenko and Mercurio(2019)]{lyashenko2019looking}
A.~Lyashenko and F.~Mercurio.
\newblock {LIBOR} replacement: a modelling framework for in-arrears term rates.
\newblock \emph{Risks}, 57--62, 2019.

\bibitem[Lyashenko and Mercurio(2020)]{lyashenko2020}
A.~Lyashenko and F.~Mercurio.
\newblock {LIBOR} replacement {II}: completing the generalized forward market
  model.
\newblock \emph{Risks}, 61--62, 2020.

\bibitem[López-Salas et~al.(2024)López-Salas, Pérez-Rodríguez, and
  Vázquez]{lopezsalas24}
J.~G. López-Salas, S.~Pérez-Rodríguez, and C.~Vázquez.
\newblock {PDE}s for pricing interest rate derivatives under the new
  generalized forward market model ({FMM}).
\newblock \emph{Computers \& Mathematics with Applications}, 169:\penalty0
  88--98, 2024.

\bibitem[Marsh and Rosenfeld(1983)]{marshrosenfeld83:CEV}
T.~A. Marsh and E.~R. Rosenfeld.
\newblock Stochastic processes for interest rates and equilibrium bond prices.
\newblock \emph{The Journal of Finance}, 38\penalty0 (2):\penalty0 635--646,
  1983.

\bibitem[Mercurio(2018)]{mercurio2018multi}
F.~Mercurio.
\newblock A simple multi-curve model for pricing {SOFR} futures and other
  derivatives, 2018.
\newblock URL \url{https://ssrn.com/abstract=3225872}.

\bibitem[Palczewski(2025)]{palczewski2025partial}
A.~Palczewski.
\newblock Partial {S}chauder estimates for unbounded solutions of the
  {K}olmogorov equation.
\newblock \emph{Journal of Elliptic and Parabolic Equations}, pages 1--27,
  2025.

\bibitem[Piazzesi(2005)]{piazzesi2005bond}
M.~Piazzesi.
\newblock Bond yields and the federal reserve.
\newblock \emph{Journal of Political Economy}, 113\penalty0 (2):\penalty0
  311--344, 2005.

\bibitem[Rubio(2011)]{Rubio2010}
G.~Rubio.
\newblock The {C}auchy-{D}irichlet problem for a class of linear parabolic
  differential equations with unbounded coefficients in an unbounded domain.
\newblock \emph{International Journal of Stochastic Analysis}, 2011\penalty0
  (1):\penalty0 469806, 2011.

\bibitem[Russo and Fabozzi(2023)]{russo2023transition}
V.~Russo and F.~J. Fabozzi.
\newblock The transition from interbank offered rates to risk-free rates:
  Evolution in pricing models for interest rate derivatives.
\newblock \emph{The Journal of Fixed Income}, 32\penalty0 (4):\penalty0 45--59,
  2023.

\bibitem[Skov and Skovmand(2021)]{skov2021dynamic}
J.~B. Skov and D.~Skovmand.
\newblock Dynamic term structure models for {SOFR} futures.
\newblock \emph{Journal of Futures Markets}, 41\penalty0 (10):\penalty0
  1520--1544, 2021.

\bibitem[Tavella and Randall(2000)]{TavellaRandall}
D.~Tavella and C.~Randall.
\newblock \emph{Pricing financial instruments: the finite difference method}.
\newblock Wiley, New York, 2000.

\end{thebibliography}
 
 \vskip3cm 

\appendix
\section{Proofs of Section \ref{affinesection}} \label{proofs}

\renewcommand{\theequation}{a.\arabic{equation}}

\subsection{Proof of  Proposition \ref{propZCB}.}

\noindent We write
$ \displaystyle P_T(t,x) = \sum_{k=0}^{K}  \exp(-a_k(t) - x \hat b_k(t)) \ind_{\{t\in [r_{k},r_{k+1 })\}}
$
or equivalently
$$a(t,T) = \sum_{k=0}^{K} a_k(t) \ind_{\{t\in [r_{k},r_{k+1})\}}   \qquad \qquad  b(t,T) = \sum_{k=0}^{K} \hat b_k(t) \ind_{\{t\in [r_{k},r_{k+1})\}}.$$ 
Next, by applying the standard procedure for affine term structure models (see, for instance, \cite{bjork:arbitragetheory}, Chapter 22), we show that equation \eqref{affinePE} holds if and only if the following system is satisfied on each interval $[r_{k}, r_{k+1})$,  subject to the appropriate terminal conditions. 
 \begin{equation}\label{eq:bk}
\left\{ \begin{array}{ll}  \hat b_k'(t) + \beta (t) \hat b_k(t) - \frac{1}{2} \delta(t)  \hat b_k^2 (t)  + 1  = 0 \\
a'_k(t)  = - \alpha \hat  b_k(t) + \frac{1}{2} \gamma(t) \hat b_k^2 (t). 
\end{array} \right.
\end{equation}

\bigskip 

\noindent
The boundary conditions can be rewritten as follows. At time $T$, the terminal condition $P_T(T, x) \equiv 1$ implies
\begin{eqnarray} \label{generalterminalcond}
\left\{\begin{array}{l}
          b(T,T)=0 \\ a(T,T) = 0.
\end{array} \right.
\end{eqnarray}

\bigskip 

\noindent At the intermediate times, if $r_k \in \cT$, the condition 
$P_T(r_{k}^-,x) = \de^{-x} P_T(r_{k}, x)$ yields  
\begin{eqnarray} 
\label{tkcond}
\left\{\begin{array}{l}
          \hat b_{k-1}(r_{k}^-)= \hat b_{k}(r_{k}) + 1 \\   a_{k-1}(r_{k}^- )= a_{k}(r_{k}).
\end{array} \right.
\end{eqnarray}

\bigskip

\noindent For $r_k \in \cS$, $r_k=s_j$ for some $j \in \{1, \ldots, M\}$, the boundary condition  
$$ P_T(r_{k}^-, x) =   \mathbb{E}[\de^{-  a_{k}(r_k) - (x + \xi_j)  \hat b_{k}(r_k )}] =  \de^{ -a_{k}(r_k) + \log \mathbb{E}[\de^{-\xi_j \hat b_{k}(r_k)}]  - x  \hat b_{k}(r_k)} $$
implies 
\begin{eqnarray} \label{skcond}
\left\{\begin{array}{l}
          \hat b_{k-1}(r_k^-)=\hat b_{k}(r_k) \\ a_{k-1}(r_k^- ) = a_{k}(r_k) - \log \mathbb{E}[\de^{-\xi_j  \hat b_{k}(r_k)}]. \end{array} \right. 
  \end{eqnarray} 

\bigskip
\noindent The first equation in system  \eqref{skcond} implies that $b$ is continuous at $r_k \in \cS$, therefore it must have the form $b(t,T) =  \sum_{n = 0}^{N} b_n(t) \ind_{\{t\in [t_{n}, \, t_{n+1})\}} $. 
The second equation in system \eqref{tkcond} implies that $a$ is continuous at $r_k \in \cT$. Exploiting a backward inductive argument, we observe that on the interval $[r_{K}, T)$, 
$$ a(t,T) = a_{K}(t) = \int_t^T \left(\alpha(u) \hat b_{K+1}(u) - \frac{1}{2} \gamma^2(u)  \hat b^2_{K+1} (u) \right) \dd u.
$$
We next move to the interval  $[r_{K-1}, r_{K})$. If $r_K =t_N $, 
\begin{eqnarray*} a(t,T)& = & a_{K-1}(t) = a_{K} (t_N) + \int_t^{t_N} \left(\alpha(u) \hat b_{K-1}(u) - \frac{1}{2} \gamma^2(u)  \hat b^2_{K-1} (u) \right) \dd u 
\\ & = &  \int_{t_N}^T\left(\alpha(u) \hat b_{K-1}(u) - \frac{1}{2} \gamma^2(u)  \hat b_{K-1} (u) \right) \dd u +  \int_t^{r_K} \left(\alpha(u) \hat b_{K-1}(u) - \frac{1}{2} \gamma^2(u)  \hat b_{K-1} (u) \right) \dd u \\&=&
\int_t^T \left(\alpha(u) b(u,T)  - \frac{1}{2} \gamma^2(u)  b^2(u, T) \right) \dd u.
\end{eqnarray*}
If $r_K =s_M $, with a similar argument we obtain  
\begin{eqnarray*} a(t,T)& = &  -\log \mathbb{E}[\de^{-\xi_M  b(s_M, T)} ]  + \int_t^T \left(\alpha(u) b(u,T)  - \frac{1}{2} \gamma^2(u)  b^2(u, T) \right) \dd u.
\end{eqnarray*}
A recursive argument yields the claim. 
  $\square $
 
\vskip1cm 

\subsection{Proof of Proposition \ref{affinecallprop}} \label{proofcall}  

The results in the classical (no-jump) case suggest that on each subinterval $[r_{k-1}, r_{k})$ the solution should take the form \eqref{affinecallpricefinal}, since under the assumption that jumps are normally distributed, the short-rate has a normal distribution as well. 
To verify that \eqref{affinePE} is satisfied, 
 we begin by computing the necessary derivatives and substituting them into the  { PDE}.
For clarity and brevity, we suppress the dependence on $x$ and $t$ when it is unambiguous.

\begin{eqnarray}
\label{der1} {\partial_t} f &=&   (\partial_t P_S)  \Phi(d_1 ) - \widehat{K} (\partial_t P_T)\Phi(d_2) +  P_S    \Phi'(d_1 ) (\partial_t d_1) - \widehat{K} P_T   \Phi'(d_2 ) (\partial_t d_2), \\   \label{der2}  
\partial_x f &=&     (\partial_x P_S)  \Phi(d_1)- \widehat{K} (\partial_x P_T)\Phi(d_2 )  + P_S   \Phi'(d_1 ) (\partial_x d_1) - \widehat{K} P_T   \Phi'(d_2 ) (\partial_x d_2),\\    \label{der3} \partial_t d_2& = & \partial_t d_1 -\sigma'_c(t),   \\ \label{der4} 
 \partial_x d_1& = &\partial_x d_2   =  \frac{1}{\sigma_c}  \left[\frac{\partial_x P_S}{P_S }   -  \frac{\partial_x P_T}{P_T }   \right] = \frac{b(t,T) - b(t,S)}{\sigma_c}, \\ \label{der5} \partial^2_{xx} d_1  &=&  \partial^2_{xx} d_2 = 0,  \\ \label{der6} 
 \partial^2_{xx} f &=&    (\partial^2_{xx} P_S) \Phi(d_1 ) - \widehat{K} (\partial^2_{xx} P_T) \Phi(d_2)     +  2   \left [ (\partial_x P_S) \Phi'(d_1 ) (\partial_x d_1) - \widehat{K} (\partial_x P_T)   \Phi'(d_2 ) (\partial_x d_2) \right]  \\ & + &   P_S    \Phi''(d_1 ) \left(\partial_x d_1\right)^2 - \widehat{K} P_T   \Phi''(d_2 ) \left(\partial_x d_2\right)^2. \nonumber 
\end{eqnarray}
Substituting \eqref{der1}, \eqref{der2} and \eqref{der6} into the first equation of system \eqref{affinePE}, and using the fact that this equation holds for  $P_T$ and $P_S$, we obtain the following equation
\begin{eqnarray} \label{caplet1}
& &   P_S   \Phi'(d_1 ) (\partial_t d_1) - \widehat{K} P_T   \Phi'(d_2 ) (\partial_t d_2) + (\alpha + \beta x) \left[P_S   \Phi'(d_1 ) (\partial_x d_1) - \widehat{K} P_T   \Phi'(d_2 ) (\partial_x d_2)\right]    \\ \nonumber & + & \sigma^2 \left[(\partial_x P_S) \Phi'(d_1 ) (\partial_x d_1) - \widehat{K} (\partial_x P_T)   \Phi'(d_2 ) (\partial_x d_2)\right] + \frac{\sigma^2}{2} \left[ P_S    \Phi''(d_1 ) \left(\partial_x d_1\right)^2 - \widehat{K} P_T  \Phi''(d_2 ) \left(\partial_x d_2\right)^2\right ] = 0. 
\end{eqnarray}

\bigskip 
 
\noindent Note that
\begin{eqnarray*}
& & \Phi'(d_1) = \frac{1}{\sqrt{2\pi} } \de^{- \frac{d_1^2}{2}}  \qquad \qquad  \qquad \Phi'(d_2) = \frac{1}{\sqrt{2\pi} } \de^{- \frac{d_1^2}{2}}  \de^{- \frac{\sigma_c^2}{2} +\sigma_c d_1} = \Phi'(d_1) \de^{- \frac{\sigma_c^2}{2} + \sigma_c d_1}    \\ & & 
 \Phi''(d_1) = -d_1  \Phi'(d_1)  \qquad \qquad  \qquad \Phi''(d_2) = -(d_1-\sigma_c) \Phi'(d_1) \de^{- \frac{\sigma_c^2}{2} + \sigma_c d_1} .
\end{eqnarray*}
\bigskip  

\noindent In addition,
$\displaystyle  - \frac{\sigma_c^2}{2} + \sigma_c d_1 = \log\left(\frac{P_S}{P_T \widehat{K}}\right)  
\qquad \mbox {implies} \qquad \quad \left\{\begin{array}{l}    \de^{- \frac{\sigma_c^2}{2} + \sigma_c d_1 }  =  \frac{P_S }{P_T \widehat{K}} \\ P_S =\widehat{K} P_T  \de^{- \frac{\sigma_c^2}{2} + \sigma_c d_1}  = \widehat{K} P_T \frac{\Phi'(d_2)}{\Phi'(d_1)}. \end{array}\right.
$

\bigskip  

\noindent 
Therefore
\begin{eqnarray*}
& & \displaystyle  P_S    \Phi'(d_1 ) (\partial_t d_1) - \widehat{K} P_T   \Phi'(d_2 ) (\partial_t d_2) = \Phi'(d_1) P_S   \sigma'_c(t)  
\\& &	
  \displaystyle   P_S   \Phi'(d_1 ) (\partial_x d_1) - \widehat{K} P_T   \Phi'(d_2 ) (\partial_x d_2)   = 0
	\\& &
	 \displaystyle    (\partial_x P_S) \Phi'(d_1 ) (\partial_x d_1) - \widehat{K} (\partial_x P_T)   \Phi'(d_2 ) (\partial_x d_2)   = P_S \Phi'(d_1 ) \frac{(b(t,T) - b(t,S))^2 }{\sigma_c(t)} 
	\\& &
	 \displaystyle    P_S   \Phi''(d_1 ) \left(\partial_x d_1\right)^2 - \widehat{K} P_T   \Phi''(d_2 ) \left(\partial_x d_2\right)^2  =  -  P_S     \Phi'(d_1 ) \frac{(b(t,T) - b(t,S))^2 }{\sigma_c(t)}  . 
	\end{eqnarray*}

\bigskip 
\noindent Equation \eqref{caplet1} becomes
\begin{eqnarray*} \label{caplet2}
& &   \Phi'(d_1) P_S   \sigma'_c(t)    +  \sigma^2 P_S \Phi'(d_1 ) \frac{(b(t,S) - b(t,T))^2 }{\sigma_c} - \frac{\sigma^2}{2}  P_S     \Phi'(d_1 ) \frac{(b(t,T) - b(t,S))^2 }{\sigma_c(t)}   = 0 ,
\end{eqnarray*}

\medskip \noindent which is satisfied if $\sigma_c$ solves the differential equation  
 \begin{equation} \label{sigmaED}
 \sigma'_c(t) + \frac{\sigma^2}{2}  \frac{(b(t,T) - b(t,S))^2 }{\sigma_c(t)} =0    
 \end{equation}
 on each subinterval $(r_{k-1}, r_k)$, subject to the appropriate boundary condition at $r_k^-$.
 
\noindent Since $b(t,S) - b(t,T) = e^{\beta(T-t)} b(T,S)$,
equation \eqref{sigmaED} simplifies to  
$$2 \sigma_c(t) \sigma'_c(t) = - \sigma^2 b^2(T,S) \de^{2\beta (T-t)},$$
which leads to   
 \begin{equation}  \label{sigmac}
 \sigma^2_c(t) = \sigma^2_c(r_k^-) + \frac{\sigma^2}{2\beta} \, b^2(T,S) \left(\de^{2\beta(T-t)} - \de^{2\beta(T-r_k)}\right).   \end{equation}

 \bigskip 
\noindent The terminal condition $ f(T, x) = H(x)$ implies  $d_1 (T, x)=d_2  (T, x) = + \infty$ if $P_S(T, x) \ge \widehat K$, and  $d_1 (T, x)=d_2  (T, x) = - \infty$ if $P_S(T, x) < \widehat K$, hence $\sigma_c(T) = 0$. Therefore,  on the interval  $[r_K, r_{K+1}) = [r_K, T)$,  we have   
 $$ \sigma^2_c(t) = b^2(T,S) \, {\frac{\sigma^2}{2\beta} \,   \left(\de^{2\beta(T-t)} - 1\right)}.
 $$
Let us now consider a generic interval $[r_{k-1}, r_k)$. 
If $r_k \in \cT$,  the terminal condition in \eqref{affinePE} implies
\begin{equation}\label{caplet_term_t}   
P_S(r_k^-) \Phi(d_1(r_k^-,x)) - \widehat{K} P_T(r_k^-) \Phi(d_2(r_k^-,x)) = \de^{-x }  \left[P_S(r_k) \Phi(d_1(r_k,x)) - \widehat{K} P_T(r_k) \Phi(d_2(r_k,x))\right].
\end{equation}	 
The boundary conditions for the ZCB price $P_{S/T}(r_k^-) = \de^{-x} P_{S/T}(r_k)$ implies that \eqref{caplet_term_t} is satisfied if $d_{1/2}(r_k^-,x) = d_{1/2}(r_k,x)$. It also implies that $P_{S}(r_k^-)/P_{T}(r_k^-) = P_{S}(r_k)/P_{T}(r_k)$, therefore it must be $ \sigma_c(r_k^-) =  \sigma_c(r_k)$,  $\sigma_c$ is continuous at $r_k \in \cT$, and, consequently,  
\begin{eqnarray*}
\sigma^2_c(t) = \sigma^2_c(r_k) +  {\sigma^2}  b^2(T,S) \int_t^{r_k} \de^{2\beta(T-t)} \, \dd u   
 = \sigma^2_c(r_{k}) + \frac{\sigma^2}{2\beta} \, b^2(T,S) \left(\de^{2\beta(T-t)} - \de^{2\beta(T-r_{k})}\right).
\end{eqnarray*}  
To find $\sigma_c(r_k^-)$, when $r_k = s_j \in \cS$,  we use the terminal condition in  \eqref{affinePE}, that implies
\begin{eqnarray} \label{capletskT}
 &  P_S(s_j^-, x) \Phi(d_1(s_j^-,x)) = \mathbb{E} \left[ P_S(s_j, x + \xi_j)  \Phi(d_1(s_j, x+\xi_j))  \right]  
\\  
 	&  \label{capletskS} P_T(s_j^-, x) \Phi(d_2(s_j^-,x)) = \mathbb{E} \left[ P_T(s_j, x + \xi_j)  \Phi(d_2(s_j, x+\xi_j))  \right].
	\end{eqnarray}
 Because of \eqref{skcond}, that is the  terminal conditions on $P_{T/S}$,  \eqref{capletskT} and   \eqref{capletskS} become
\begin{eqnarray} \label{d1skT}
&   \mathbb{E} \left[ e^{-\xi_j b(s_j, S)} \right] \Phi(d_1(s_j^-,x)) = \mathbb{E} \left[ e^{-\xi_j b(s_j, S)}  \Phi(d_1(s_j, x+\xi_j))  \right]
\\  
 	&  \label{d2skS}  \mathbb{E} \left[ \de^{-\xi_j b(s_j, T)} \right]  \Phi(d_2(s_j^-,x)) = \mathbb{E} \left[ \de^{-\xi_j b(s_j, T)}  \Phi(d_2(s_j, x+\xi_j))  \right] .
 	\end{eqnarray}

\bigskip 

\noindent Since $\xi_j$ has a Gaussian distribution with mean $m_j$ and variance $\gamma^2_j$, we have   
 \begin{equation}
    \mathbb{E} \left[ \de^{-\xi_j b(s_j, S)} \right] = \de^{-m_j b(s_j, S) + \gamma^2_j b^2(s_j, S)/2} = \de^{-m_j b(s_j, S) + \gamma^2_j b^2(s_j, S)/2} .
 \end{equation}
 
 \medskip \noindent To calculate the right-hand side of \eqref{d1skT}, 
 we   write $\xi_j= m_j + \gamma_j Z$ where $ Z \sim \mathcal{N} (0,1)$ so that 
 \begin{equation}\label{rhs}
  \mathbb{E} \left[ \de^{-\xi_j b(s_j, S)}  \Phi(d_1(s_j, x+\xi_j))  \right]  
 = \de^{- m_j  b(s_j,S) }   \mathbb{E} \left[ \de^{- \gamma_j b(s_j, S) Z } \Phi(d_1(s_j, x+ \xi_j)\right]  .
 \end{equation}
 Therefore \eqref{d1skT} is equivalent to
\begin{equation} \label{d1skT2}
   \Phi(d_1(s_j^-,x)) = \de^{-\gamma^2_j b^2(s_j, S)/2}\mathbb{E} \left[ \de^{-\gamma_j b(s_j, S) Z}  \Phi(d_1(s_j, x+ m_j + \gamma_j Z))  \right]  .
     \end{equation}

\noindent 
Moreover,  
\begin{eqnarray*} d_1(s_j, x+ \xi_j))  = d_1(s_j, x) + \frac{ b(s_j,T) - b(s_j,S)}{\sigma_c(s_j)} \, \xi_j
=  d_1(s_j, x) - \frac{\de^{\beta (T-s_j)} b(T,S)}{\sigma_c(s_j)} \, \xi_j = D^1_j - \Sigma_j Z
\end{eqnarray*}
where we set 
$$D^1_j = d_1(s_j, x) - \frac{\de^{\beta (T-s_j)} b(T,S)}{\sigma_c(s_j)} \, m_j, \qquad \qquad \Sigma_j= \frac{\gamma_j \de^{\beta (T-s_j)} b(T,S)}{\sigma_c(s_j)}.$$
Denoting by $\varphi_{\mu, \sigma^2}$ the density function of $\mathcal{N}(\mu, \sigma^2)$ and applying Fubini-Tonelli, the expectation in the right-hand side of \eqref{d1skT2}  can be calculated as
\begin{eqnarray}
 	\nonumber 	 \mathbb{E} \left[ \de^{-\gamma_j b(s_j, S)  Z } \Phi(d_1(s_j, x+\xi_j))\right] & = & \int_{-\infty}^{+\infty} \de^{-\gamma_j b(s_j, S)  z}  \varphi_{0,1}(z)   \int_{-\infty}^{D^1_j -\Sigma_j z}  \varphi_{0,1} ( y ) \, \dd y \, \dd z \\  & = &  \int_{-\infty}^{+\infty} \varphi_{0,1}(y) \int_{-\infty}^{\frac{D^1_j-y}{\Sigma_j}}  \de^{-\gamma_j b(s_j, S) z}  \varphi_{0,1} (z) \, \dd z \, \dd y .\label{Fubini1}
 \end{eqnarray} 

\medskip\noindent For simplicity, denote $\tilde \gamma = \gamma_j b(s_j, T)$ and observe that $\de^{-\tilde \gamma^2/2 -\tilde \gamma z}  \varphi_{0,1} (z) = \varphi_{-\tilde \gamma,1} (z) $.   Then, we can write the right-hand side of    \eqref{d1skT2}  
  as 
  \begin{eqnarray*}   \int_{-\infty}^{+\infty} \varphi_{0,1}(y) \int_{-\infty}^{\frac{D^1_j-y}{\Sigma_j}}    \varphi_{-\tilde \gamma,1} (z) \, \dd z \, \dd y  &=&  \mathbb{Q} \left(\tilde Z   \le  \frac{D^1_j -Y}{\Sigma_j} \right)  =      \mathbb{Q} \left(\Sigma_j \tilde Z  + Y  \le  D^1_j  \right) ,
 \end{eqnarray*}
 where $\tilde Z \sim \mathcal{N}( -\tilde \gamma ,  1)$, $Y\sim \mathcal{N}( 0, 1)$. Because they are independent,  $\Sigma_j \tilde Z  + Y \sim \mathcal{N}( -\Sigma_j \tilde \gamma , \Sigma_j^2 +1)$, hence 
 $$ \mathbb{Q} \left(\Sigma_j \tilde Z  + Y  \le  D^1_j  \right)   =     \Phi  \left(   \frac{D^1_j  + \Sigma_j \tilde \gamma}{\sqrt{\Sigma_j^2 +1}} \right).
 $$ 
Equation \eqref{d1skT2} is therefore equivalent to
\begin{equation*}
\Phi(d_1(s_j^-,x)) =      \Phi  \left(   \frac{D^1_j  +  \Sigma_j \tilde \gamma}{\sqrt{\Sigma_j^2 +1}} \right), 
\end{equation*}
that is, 
\begin{eqnarray*}
 d_1(s^-_j,x) &=&        \frac{D^1_j  +  \Sigma_j \gamma_j b(s_j, S)}{\sqrt{\Sigma_j^2 +1}}. 
 \end{eqnarray*}
By applying the same procedure, from equation \eqref{d2skS} we obtain 
$$ d_2(s^-_j,x) =         \frac{D^2_j  + \Sigma_j\gamma_j b(s_j, T)}{\sqrt{\Sigma_j^2 +1}},
$$
where $D^2_j=  d_2(s_j, x) - \frac{\de^{\beta (T-s_j)} b(T,S)}{\sigma_c(s_j)} \, m_j = D^1_j - \sigma_c(s_j).$
It follows that
\begin{gather*}
\sigma_c(s^-_j) = d_1(s^-_j,x) - d_2(s^-_j,x) =  
\frac{\sigma_c(s_j)  +  \Sigma_j (\gamma_j b(s_j, S) -\gamma_j b(s_j, T))}{\sqrt{\Sigma_j^2 +1}}.  
\end{gather*}
Since $(\gamma_j b(s_j, S) -\gamma_j b(s_j, T)) = \gamma_j \de^{\beta (T-s_j)} b(T,S) = \sigma_c(s_j) \Sigma_j$
we have that  
\begin{eqnarray*}
\sigma_c(s^-_j)& = &   
\frac{\sigma_c(s_j)  +  \sigma_c(s_j) \Sigma^2_j}{\sqrt{\Sigma_j^2 +1}}  =  \sigma_c(s_j) \sqrt{\Sigma_j^2 +1}
\end{eqnarray*}
and, as a consequence
$$\sigma^2_c(s^-_j) = \sigma^2_c(s_j) (\Sigma^2_j + 1) = \sigma_c^2 (s_j) + \gamma^2_j \de^{2\beta (T-s_j)} b^2(T,S).
$$
A recursive application of this equality yields  
\begin{equation*}  
	\sigma^2_c (t) = b^2(T,S)  \left(   {\frac{\sigma^2}{2\beta} \left(\de^{2\beta (T-t)}-1 \right)  +  \sum_{s_j \in \cS, t < s_j < T}\gamma_j^2 \de^{2\beta (T-s_j)} }  
	\right). \qquad \qquad  \qquad \qquad \qquad \square
     \end{equation*}
 
\bigskip 
 
\renewcommand{\theequation}{b.\arabic{equation}}

\section{Mathematical and numerical insights} 

\subsection{Green's function of the pricing PDE in the Vasicek model}\label{AppGreen}

We prove that the fundamental solution over $\mathbb{R}$ associated with the PDE \eqref{VasPDE} in the Vasicek model is given by
\begin{equation}\label{GF}
G(t,s;x,\xi)= \frac{C_1(t,s;\xi)C_2(t,s;\xi)}{\sqrt{2\pi}\Sigma(t,s)} \exp \left(-\frac{(\xi-\mu(t,s;x))^2}{2\Sigma^2(t,s)}\right)=G(s-t;x,\xi),\qquad s>t,
\end{equation}
where
$$\Sigma^2(t,s)=-\frac{\sigma^2}{2\beta}(1-\de^{2\beta (s-t)}), \ \ \
\mu(t,s;x)=x \de^{\beta (s-t)} -\frac{\alpha}{\beta}(1-\de^{\beta (s-t)})+\frac{\sigma^2}{2\beta^2}(\de^{-\beta (s-t)}+\de^{\beta (s-t)}-2),$$
$$C_1(t,s;\xi)=\exp \left( -\frac{\sigma^2 \de^{-2\beta (s-t)}-4\xi\beta^2\de^{-\beta (s-t)}-4 \sigma^2 \de^{-\beta (s-t)} -4\alpha \beta \de^{-\beta (s-t)}}{4\beta^3}\right),$$
$$C_2(t,s;\xi)=\exp \left( \left(\frac{\sigma^2}{2\beta^2}+\frac{\alpha}{\beta}\right) (s-t) -\frac{3\sigma^2}{4\beta^3}-\frac{\xi}{\beta}-\frac{\alpha}{\beta^2}\right).$$
This expression can be derived by following the procedure proposed in \cite{Lo2013}, which is based on the algebraic theory of Lie groups\footnote{In \cite{Buttler1996} an expression of the fundamental solution is obtained using Hermite polynomials and their properties.}.

\begin{proof}
With cumbersome calculations of derivatives, we can verify that
$G(t,s;x,\xi)$ is solution to the backward Kolmogorov equation in \eqref{VasPDE} for each $(\xi,t)$
\begin{equation}\label{eq:Vas}
    \partial_t G(t,s;x,\xi) + (\alpha + \beta x)  \partial_x G(t,s;x,\xi) + \frac{1}{2} \sigma^2\partial^2_{xx} G(t,s;x,\xi) - x\,G(t,s;x,\xi) =0\,.
\end{equation}
and of the forward adjoint equation for each $(x,s)$, with $s\geq t$,
\begin{equation}\label{EqKolm_adj}
-\partial_t G(t,s;x,\xi) - (\alpha + \beta \xi)  \partial_\xi G(t,s;x,\xi) + \frac{1}{2} \sigma^2\partial^2_{\xi\xi} G(t,s;x,\xi) - (\beta+\xi)G(t,s;x,\xi) =0\,.
\end{equation}
Moreover, for each compactly supported smooth function
$\phi\in\mathcal{C}_0^\infty(\mathbb{R})$, with the usual change of variable
$y=\displaystyle\frac{\xi-\mu(t,s;x)}{\Sigma(t,s)}$, it holds that
$$
\begin{array}{rcl}
  \displaystyle\int_{-\infty}^{\infty} G(t,s;x,\xi) \phi(\xi)\, \dd\xi & = & \displaystyle\int_{-\infty}^{\infty}\frac{C_1(t,s;\xi)C_2(t,s;\xi)}{\sqrt{2\pi}\Sigma(t,s)} \exp \left(-\frac{(\xi-\mu(t,s;x))^2}{2\Sigma^2(t,s)}\right)\phi(\xi) \, \dd\xi.
\end{array}
$$
$$
\begin{array}{rcl}
   & = & \displaystyle\int_{-\infty}^{\infty}\frac{C_1(t,s;y)C_2(t,s;y)}{\sqrt{2\pi}\Sigma(t,s)} \de^ {-\frac{y^2}{2}}\phi(y
\Sigma(t,s)+\mu(t,s;x))\Sigma(t,s)\, \dd y \\
   & = &  \displaystyle\frac{1}{\sqrt{2\pi}} \int_{-\infty}^{\infty}f(t,s;x,y) \de^ {-\frac{y^2}{2}} \, \dd y=:\mathcal{I}(t,s;x),
\end{array}
$$
having defined
$$
f(t,s;x,y):=\big\{C_1(t,s;y)C_2(t,s;y)\phi(y
\Sigma(t,s)+\mu(t,s;x))\big\}
$$
so that
$$
\begin{array}{rcl}
 \displaystyle\frac{1}{\sqrt{2\pi}}\min_{y\in \mathbb{R}}f(t,s;y,x) \int_{-\infty}^{\infty} \de^ {-\frac{y^2}{2}} \, \dd y  & \leq \mathcal{I}(t,s;x) \leq &
 \displaystyle\frac{1}{\sqrt{2\pi}}\max_{y\in \mathbb{R}}f(t,s;y,x) \int_{-\infty}^{\infty} \de^ {-\frac{y^2}{2}} \, \dd y\,.
\end{array}
$$
Moreover, for $(s-t)\rightarrow 0$, the following limits hold
$$
\Sigma(t,s)\rightarrow 0\qquad \mu(t,s;x)\rightarrow x\qquad C_1(t,s;y)C_2(t,s;y)\rightarrow 1\qquad f(t,s;x,y)\rightarrow x.
$$
Hence
$$
\lim_{(s-t)\rightarrow 0}\mathcal{I}(t,s;x)=\phi(x)=\int_{-\infty}^{\infty}  \phi(\xi)\delta_{x}(\xi)\, \dd\xi
$$
and we can conclude that the limit of $G(t,s;x,\xi)$ for $(s-t)\rightarrow
0$, in distributional sense, is $\delta_{x}(\xi)$.
\end{proof}

\subsection{Localization of the computational domain}\label{Localization}

In this section, we focus on a possible strategy for localizing the PDE system (\ref{eq:lastPDE})-(\ref{finalSystem}) to a suitable bounded computational domain, assuming that we know the Green's function associated with the PDE. 
The basic idea is as follows: we fix a region of {observation} \([x_{\text{min}}, x_{\text{max}}]\) and expand it to a wider domain \([\bar{A}, \bar{\bar{A}}]\) such that:
\begin{enumerate}
    \item for points $x$ within the region of {observation} and for any terminal condition $g$, the integral in equation (\ref{solG}) is accurately approximated by an integral over \([\bar{A}, \bar{\bar{A}}]\)
\begin{equation}\label{eq:val_atteso}
 f(t,x)=\int_{\mathbb{R}} G(t,r_{k+1};x,\xi) g(\xi) \, \dd\xi\approx \int_{\bar{A}}^{\bar{\bar{A}}} G(t,r_{k+1};x,\xi) g(\xi) \, \dd\xi;
\end{equation}
\item in the presence of a jump on the  { RFR} at $r_k \in {\cal S}$, for points $x \in [x_{\text{min}}, x_{\text{max}}]$, the support of ${Q}_{m(k)}$ is mostly contained within \([\bar{A}, \bar{\bar{A}}]\) and hence
\begin{equation}\label{eq:int_salti}
f(r_k^-,x)=\int_{\mathbb{R}} f(r_k,x+z) \, {Q}_{m(k)}(dz)=\int_{\mathbb{R}} f(r_k,x+z)\varphi_{m(k)}(z)\, \dd z\approx
\int_{\bar A}^{\bar{\bar{A}}} f(r_k,\xi)\varphi_{m(k)}(\xi-x)\, \dd\xi\,.
\end{equation}
\end{enumerate}
Although a heuristic approach is generally possible, in the specific case of the Vasicek model, we have developed a systematic strategy to determine $\bar{A}$ and $\bar{\bar{A}}$.

The first task, outlined in equation \eqref{eq:val_atteso}, can be addressed by applying the following lemma.
\begin{lemma} The Green's function \eqref{GF} related to the Vasicek PDE \eqref{VasPDE} satisfies
$$
\int_\mathbb{R}G(t,T;x,\xi)\de^{\frac{\xi-x}{\beta}}\, \dd\xi=\de^{\left(\frac{\sigma^2}{2\beta^2}+\frac{\alpha}{\beta}\right)(T-t)}.$$ 
\end{lemma} 
\begin{proof}
From \eqref{solG}
$$
f(t,x)=\int_{\mathbb{R}} G(t,r_{k+1};x,\xi) f(r_{k+1}^-,\xi) \, \dd\xi.
$$
Considering the change of variable suggested at page \pageref{eq:Vas_cambio}, i.e., $v(t,x)=f(t,x)\de^{-x/\beta}$, we get
$$
v(t,x)\de^{x/\beta}=\int_{\mathbb{R}} G(t,r_{k+1};x,\xi) v(r_{k+1}^-,\xi)\de^{\xi/\beta} \, \dd\xi,
$$
from which it follows that
$G(t,T;x,\xi)\de^{\frac{\xi-x}{\beta}}$ is the fundamental solution of equation 
\eqref{eq:Vas_cambio}. Then $w(t,x)= v(t,x)\de^{-\left( \frac{\sigma^2}{2\beta^2} + \frac{\alpha}{\beta} \right) (T-t)}$
satisfies
$$
\partial_t w(t,x) + \dfrac{\sigma^2}{2} \partial_{x x}^2 w(t,x) + \left(\alpha + \beta x + \dfrac{\sigma^2}{\beta}\right) \partial_{x} w(t,x)=0,  
$$
whose fundamental solution $\displaystyle G(t,T;x,\xi)\de^{\frac{\xi-x}{\beta}}
\de^{-\left( \frac{\sigma^2}{2\beta^2} + \frac{\alpha}{\beta} \right)(T-t)}$ is known from Kolmogorov theorem to be a transition probability density function. Hence
\begin{equation*}
\int_\mathbb{R}G(t,T;x,\xi)\de^{\frac{\xi-x}{\beta}}
\de^{-\left( \frac{\sigma^2}{2\beta^2} + \frac{\alpha}{\beta} \right)(T-t)}\, \dd\xi=1\,. \qedhere
\end{equation*}
\end{proof}
\noindent We can therefore use a root finding algorithm and compute $M$ such that
$$\int_{x_{\min}-M}^{x_{\max}+M}G(t,T;x,\xi)\de^{\frac{\xi-x}{\beta}} \, \dd\xi\approx \de^{\left(\frac{\sigma^2}{2\beta^2}+\frac{\alpha}{\beta}\right)(T-t)}$$ holds for every $x\in[x_{\min},x_{\max}]$, at a desired accuracy.

Then, to achieve the second task  \eqref{eq:int_salti}, in the case of normally distributed jumps on the short-rate with density $\varphi_{m(k)}$, we search for $\bar{M}$ such that
$$\int_{x_{\min}-\bar{M}}^{x_{\max}+\bar{M}}\varphi_{m(k)}(\xi-x) \,\dd\xi\approx 1, \qquad \forall x\in[x_{\min},x_{\max}],$$
at a desired accuracy, concluding that $\bar{\bar{A}}=x_{\max}+\max\{\bar{M},M\}$ and $\bar{A}=x_{\min}-\max\{\bar{M},M\}$.

\bigskip
\begin{remark}
    In the case of discretely distributed jumps, it is sufficient to enlarge the region of {observation} to cover the interval $[x_{\min} - 3m_j, x_{\max} + 3m_j]$, which is likely already contained within the computational domain $[\bar{A}, \bar{\bar{A}}]$ determined by condition (\ref{eq:val_atteso}).
\end{remark}

\end{document}